\documentclass[11pt,letterpaper,superscriptaddress,nofootinbib]{revtex4}
\usepackage{fullpage,dsfont,amsmath,amsthm,amssymb,amscd}
\usepackage{stmaryrd,bbold,bussproofs}
\usepackage{tabularx,multirow,accents,fontenc}

\usepackage{cmll}
\usepackage{enumitem}

 \normalfont

\usepackage{hyperref,cleveref}
\usepackage{color,graphicx,array,tikz}
\usetikzlibrary{cd}

\DeclareFontFamily{U}{mathb}{\hyphenchar\font45}
\DeclareFontShape{U}{mathb}{m}{n}{
      <5> <6> <7> <8> <9> <10> gen * mathb
      <10.95> mathb10 <12> <14.4> <17.28> <20.74> <24.88> mathb12
      }{}
\DeclareSymbolFont{mathb}{U}{mathb}{m}{n}

\newtheorem{theorem}{Theorem}
\newtheorem{lemma}[theorem]{Lemma}
\newtheorem{proposition}[theorem]{Proposition}
\newtheorem{corollary}[theorem]{Corollary}

\newtheorem*{conjecture*}{Conjecture}

\theoremstyle{definition}
\newtheorem{definition}[theorem]{Definition}

\newcounter{example}
\newenvironment{example}{\refstepcounter{example}\par\medskip\noindent\textbf{Example~\theexample.}}{\hfill$\triangle$\par\medskip}
\newenvironment{example*}{\par\medskip\noindent\textbf{Example.}}{\hfill$\triangle$\par\medskip}

\newcounter{pfcounter}
\renewcommand{\thepfcounter}{\arabic{pfcounter}}

\newcounter{prooftree}

\newcommand{\pfref}[1]{[{\ref{#1}}]}

\newcommand{\id}{\ensuremath \mathrm{id}}%

\newcommand{\system}[1]{\ensuremath \mathbb{#1}}

\newcommand{\infers}{\vdash}
\newcommand{\forces}{\Vdash}
\newcommand{\manifests}{\vDash}
\newcommand{\frames}[1]{\ensuremath{\left\llbracket{#1}\right\rrbracket}}
\newcommand{\rulename}[1]{\RightLabel{\scriptsize{#1}}}

\newcommand{\pfline}{\; \refstepcounter{prooftree}[\alph{prooftree}]} 
\newcommand{\scpf}{\setcounter{prooftree}{0}}
\newcommand{\ruleId}[1]{\AxiomC{}\rulename{Id}\UnaryInfC{#1 $\infers$ #1}}
\newcommand{\ruleid}[1]{\AxiomC{}\rulename{Id}\UnaryInfC{#1 $\infers$ #1}}

\newcommand{\multrue}{\ensuremath{\mathbb{1}}}
\newcommand{\mt}{\multrue}
\newcommand{\RI}{R-\mt}
\newcommand{\LI}{L-\mt}

\newcommand{\muland}{\ensuremath{\otimes}}
\newcommand{\ma}{\muland}
\newcommand{\RO}{R-$\ma$}
\newcommand{\LO}{L-$\ma$}

\newcommand{\addand}{\ensuremath \with}

\newcommand{\sysT}{\system{T}}

\newcommand{\monoid}{\ensuremath\cdot}

\newcommand{\algM}{\ensuremath \mathcal{M}}
\newcommand{\mon}{\monoid}
\newcommand{\monM}{\system{M}}

\newcommand{\category}[1]{\ensuremath{\mathsf{#1}}}
\renewcommand{\hom}{\mathrm{Hom}}
\newcommand{\monoidal}{\ensuremath\boxtimes}

\newcommand{\catC}{\category{C}}

\DeclareMathSymbol{\boxcirc}{3}{mathb}{"65}

\definecolor{greenn}{rgb}{0,0.8,0.2}
\definecolor{bluue}{rgb}{0.3,0,0.7}

\newcommand{\bnote}[1]{}

\definecolor{edcolor}{rgb}{0.2,0.5,0.3}

\begin{document}

\title{Mathematical methods for resource-based type theories}

\author{Aarthi~Sundaram}\email{aarthims@umd.edu}
\affiliation{Joint Center for Quantum Information and Computer Science, University of Maryland, College Park}

\author{Brad~Lackey}\email{bclackey@umd.edu}
\affiliation{Joint Center for Quantum Information and Computer Science, University of Maryland, College Park}
\affiliation{Institute for Advanced Computer Studies, University of Maryland, College Park}
\affiliation{Departments of Computer Science and Mathematics, University of Maryland, College Park}
\affiliation{Quantum Architectures and Computation Group (QuArC), Microsoft, Redmond}

\begin{abstract}
With the wide range of quantum programming languages on offer now, efficient program verification and type checking for these languages presents a challenge -- especially when classical debugging techniques may affect the states in a quantum program. In this work, we make progress towards a program verification approach using the formalism of operational quantum mechanics and resource theories. We present a logical framework that captures two mathematical approaches to resource theory based on monoids (algebraic) and monoidal categories (categorical). We develop the syntax of this framework as an intuitionistic sequent calculus, and prove soundness and completeness of an algebraic and categorical semantics that recover these approaches. We also provide a cut-elimination theorem, normal form, and analogue of Lambek's lifting theorem for polynomial systems over the logics. Using these approaches along with the Curry-Howard-Lambek correspondence for programs, proofs and categories, this work lays the mathematical groundwork for a type checker for some resource theory based frameworks, with the possibility of extending it other quantum programming languages. 
\end{abstract}

\maketitle

\tableofcontents

\section{Introduction}

In recent years we have seen an increasing number of practical quantum programming languages, from the high-level ones like Quipper \cite{green2013quipper} and QWIRE \cite{paykin2017qwire} to those that are hardware/simulator specific like Q\# \cite{svore2018q}, Quil \cite{smith2016practical} and QISKit \cite{qiskit2018}. These afford users the ability to implement ever larger quantum circuits, at which point the question of checking that such quantum programs behave as intended takes prominence. Classical debugging techniques generally involve observing the system state during program execution to help deduce the programming error; however, characterizing the state of a quantum register is generally intractable for all but the smallest programs. The other option is to verify the program by ascertaining its correctness in a formal model. For instance QWIRE uses density matrix and circuit formalisms \cite{feng2015qpmc} to verify a program, but this leads to similar scalability issue. In this work we do not attempt to provide methods for full program verification, but focus on the more limited task of type checking. Our approach uses the formalism of operational quantum mechanics via a resource theory framework \cite{coecke2016mathematical} and categorical quantum mechanics \cite{abramsky2004categorical}.

Resource-based reasoning is not new, the most popular being separation logic/bunched implications \cite{burstall1972some, ohearn1999logic, reynolds2002separation, pym2004possible} and linear logic/geometry of interactions \cite{girard1987linear, seely1987linear, girard1989geometry, abramsky1992new}. In computer science, the former has been used extensively for concurrence \cite{ohearn2007resources} while the later has proven more popular in quantum information \cite{selinger2006lambda}. Quantum thermodynamics has a long history as a resource theory, see for example \cite{brandao2013resource, horodecki2013fundamental, brandao2015second, faist2015minimal, gour2015resource, lostaglio2015quantum, lostaglio2015description, narasimhachar2015low, faist2018fundamental}. Similarly, many ideas from quantum foundations have been recast as resource theories including purity \cite{horodecki2003reversible, streltsov2018maximal}, entanglement \cite{groisman2005quantum, brandao2008entanglement, horodecki2009quantum, chitambar2016relating, streltsov2017colloquium}, coherence \cite{baumgratz2014quantifying, chitambar2016relating, chitambar2016comparison, marvian2016quantify, napoli2016robustness, winter2016operational, streltsov2017colloquium}, contextuality \cite{grudka2014quantifying, horodecki2015axiomatic, ahmadi2018quantification, amaral2018noncontextual}, and nonlocality \cite{devincente2014nonlocality, horodecki2015axiomatic}. Resource theory frameworks have even been developed for asymmetry \cite{gour2008resource, toloui2012simulating, marvian2014asymmetry, wakakuwa2017symmetrizing} and magic states \cite{veitch2014resource, howard2017application, ahmadi2018quantification}. Foundational works on quantum information as a general resource theory such as \cite{devetak2008resource, horodecki2013fundamental, brandao2015reversible, gour2017quantum, chitambar2018quantum} have led to mathematical frameworks for abstract resource theories as monoids \cite{fritz2015resource} or as monoidal categories \cite{coecke2013resource, coecke2016mathematical, marsden2018quantitative}.

Beyond resource theory, the ``categorification'' of quantum foundations \cite{abramsky2009categorical, baez2006quantum, chiribella2011informational, coecke2016generalised, tull2016operational}, semantics \cite{coecke2016categories, dixon2009graphical, hines2013types, selinger2010survey}, protocols \cite{abramsky2004categorical, vicary2012higher}, and computation \cite{lee2015computation, lee2017oracles} has led to quantum information being treated as a form of generalized probability theory \cite{chiribella2016operational}. From this perspective graphical languages for quantum computation \cite{selinger2010survey, coecke2017picturing} have arisen where the monoidal category \cite{maclane1963natural} forms the centerpiece. An intriguing aspect of this framework is its flexibility for designing programs and proofs. For example, a diagrammatic approach was recently used to show the self-testing property of a multipartite quantum state~\cite{breiner2018parallel}\bnote{check this reference}. 

To validate a program at the level of its types, there are two broad tasks to be accomplished: (i) check that the processes take in the correct types/resources and all type conversions are valid, and (ii) ensure that the processes compose in a way that forms a monoidal category. In this work, we present a type theory for program verification by building the minimal logical fragment that captures axiomatic resource theories. We refer to this logic as $\sysT$ for \emph{tensor}, which contains only a nullary tautology $\multrue$ and a multiplicative conjunction $\muland$. The term ``multiplicative'' refers to the position of $\sysT$ within linear logic or the logic of bunched implications. Opposed to this is ``additive'' logics, where propositions represent properties rather than resources; these are better suited for classical systems. We do not propose an implication, which would recover closed categories in our semantics. Rather we follow a stricter Curry-Howard-Lambek correspondence: objects correspond to terms while morphisms correspond to proofs in our logic. This is our first step towards developing strong quantum type checkers for quantum programming languages using categorical quantum mechanics.

\paragraph*{\bf The deductive system.}  The system $\sysT$, containing only the multiplicative conjunction $\otimes$ as a connective, is a fragment of multiplicative intuitionistic linear logic \cite{girard1987linear}. Intuitively, this connective is analogous to the tensor product in a monoidal category. In Section \ref{section:syntax} we formally develop the syntax of $\sysT$ as an intuitionistic logic using the sequent calculus (see Table \ref{system-T:rules-table}). The various transformations that define the equivalence classes for logical proofs is discussed in Section \ref{section:transformations}. In Section \ref{section:cut-elimination} we prove a cut-elimination theorem~\cite{gentzen1964investigations}: if an inference in system $\sysT$ has a proof, then there is an \emph{equivalent} proof that does not use the cut rule.

A critical feature of resource theories is that some resources as deemed to be freely available or freely disposable. Logically, this translates to a term that can be arbitrarily added to the consequent or antecedent of an inference $\Gamma \infers A$. Here, $A$ denotes an atomic proposition and $\Gamma$ is a sequent i.e., a (possibly infinite) set of propositions. To capture a resource being \emph{free}, we introduce polynomial systems in Section \ref{section:polynomials} which are created by adding rules to $\sysT$. Namely, when the resource represented by term $X$ is freely available, we get the system $\sysT[X]$ that has the additional rule: $\infers X$. When $X$ is freely disposable, we get the polynomial system $\sysT[\bar{X}]$ by having the additional rule $X \infers \multrue$.

\begin{example}\label{example:introduction}
Perhaps the simplest illustrative example is the system $\sysT$ with two atomic propositions $C,Q$ representing the types of a classical and quantum bit respectively. The term $\mt$ is the unit type. One can form types of registers of classical or quantum bits as $C^{\ma n}$ or $Q^{\ma m}$. In the system $\sysT[C]$ classical information can be freely introduced; we will see in Section \ref{section:polynomials} below that ``cloning'' $C \infers C\ma C$ is a theorem of $\sysT[C]$. Opposed to this, in the system $\sysT[\overline{C}]$ classical information can be freely disposed. One lifting theorem proved below, Proposition \ref{systemT:disposable-lifting-theorem}, states that each typing relation $\Gamma \infers A$ provable in $\sysT[\overline{C}]$ corresponds to a relation $\Gamma \infers A\otimes C^{\ma n}$ provable in $\sysT$ where any erased classical information is retained in an ancillary register. In the logical system $\sysT[C,\overline{C}]$, classical information is free in the usual resource-theoretic sense.

Additionally, we can form a toy model of coherence by adding the rule $Q \infers C$, capturing the ability to freely decohere quantum information into classical information in the basis implicitly represented by $C$. In Example \ref{example:coherence} below we will consider a more detailed example of this form were we have a continuous resource monotone.
\end{example}

We provide two general notions of semantics for $\sysT$ in this work.

\paragraph*{\bf 1. Algebraic semantics.} The algebraic semantics of a logic involves mapping each term into an algebraic structure so that the connectives of the logic are realized by the operations of the algebra.
A proof of an inference in the logic then corresponds to a relation between elements in the algebra. 
An algebraic model $\algM$ for $\sysT$ is a commutative ordered monoid\footnote{The monoid operation $\mon$ and order $\leq$ are compatible in that whenever $r \leq s$ and $x \leq y$ then $r \mon x \leq s \mon y$.} $\monM$ along with a valuation function $v: \Phi \mapsto \monM$ which maps the conjunction of terms to their product in the monoid. Intuitively, this valuation encodes the resources associated to a term. We use a forcing relation $m \forces A$ (read, $m \in \monM$ \emph{forces} $A \in \sysT$) to represent instances of type $A$ having sufficient resources to instantiate $m$. We require forcing to respect the $\muland$ connective by imposing the following rule: $m \forces A \muland B$ if and only if there exists $n_1, n_2 \in \monM$ with $m = n_1 \mon n_2$ such that $n_1 \forces A$ and $n_2 \forces B$.

In Section \ref{section:algebra}, we use this forcing relation to define semantic entailment over a requirement $\Gamma \manifests_m A$ and prove the soundness and completeness of the algebraic semantics of $\sysT$. In other words, we show that $\Gamma \infers A$ has a proof in $\sysT$ if and only if $\Gamma \manifests_m A$ for all $m \in \monM$ and models $\algM$. Note that by developing semantics through ordered monoids we recover the formalism of \cite{fritz2015resource}.

\paragraph*{\bf 2. Categorical Semantics.} Given a language, one constructs a \emph{syntactic category} whose objects are the types in the language, and morphisms are the functions defined between the various types that satisfy the relations provable in the language. The syntactic category provided by $\sysT$ is a symmetric monoidal category\footnote{A symmetric monoidal category is one equipped with a symmetric bifunctor that is associative up to a natural transformation and an object $I$ that is both a left and right identity for the bifunctor.} $\catC$. Every term becomes an object of $\catC$ where the identity object $I = \multrue$. A proof of an inference $\Pi: A \infers B$ becomes a morphism $\frames{\Pi} \in \hom(A, B)$, and the cut rule is used to define the composition of morphisms. The logical connective $\muland$ becomes the symmetric bifunctor $\monoidal$ where $\frames{\Pi_1} \monoidal \frames{\Pi_2}$ is given by combining proofs $\Pi_1$ and $\Pi_2$ using the logical rules for the $\otimes$ connective (see \RO\ and \LO\ rules in Table \ref{system-T:rules-table}). 

The bifunctor $\monoidal$, and related natural transformations, satisfy certain commutative diagrams (like the hexagon rule and triangle rule~\cite{maclane1963natural}) that induce an equivalence relation on proofs: $\Pi \sim \Pi'$ when $\frames{\Pi_1} = \frames{\Pi_2}$. Using proof transformations that maintain the equivalence of inferences in $\sysT$ we show in Section \ref{section:category} all these diagrams in $\catC$ can be satisfied. Finally, we prove coherence in that the unique functor from the syntactic category is fully faithful, thereby recovering the formalism of \cite{coecke2016mathematical}.

\paragraph*{\bf Related Work \& Future Directions.} A Curry-Howard-Lambek correspondence for quantum logic has a significant literature; some of the most relevant constructions for our work involve dagger-closed categories \cite{coecke2010categories} and the quantum typed $\lambda$-calculus developed in \cite{selinger2006lambda}. The latter is a quantum analogue of Lambek's $\lambda$-calculus for classical computing. Certainly complicated categorical semantic constructions and strongly typed quantum $\lambda$-calculus translate to more expressive logics and get closer to the complete power of linear logic. Equally certain is that implementing type checkers for them is a correspondingly bigger challenge~\cite{selinger2018challenges}.

QWIRE~\cite{paykin2017qwire} is an embedded language in the proof assistant Coq and is enhanced with density operator denotational semantics that works well for small circuits. QPMC~\cite{feng2015qpmc} is a model checker using density operator and quantum channel denotational semantics that has been used for some Quipper programs~\cite{anticoli2016towards}. Proto-quipper~\cite{ross2015algebraic, rios2017categorical} tries to bridge the theoretical-practical gap in quantum type theory by formalizing some aspects of Quipper and builds on categorical semantics. 

In contrast, our approach takes the simplest possible logic that can still capture non-trivial aspects of resource theories and study its capabilities. We aim to extend the type theory to more expressive logics with more functionality in the future. In one direction, we look to build a practical type checking tool with strong theoretical foundations that could also be independently incorporated as a separate tool into hardware-specific languages. In another direction, we would like to understand how to enrich the logic to add classical control imperatives like if-then-else which system $\sysT$ cannot express. Currently, some classical functionality akin to the additive conjunction ``$\addand$'' can be recovered by setting a classical bit $C$ to be \emph{free} -- freely available and freely disposable -- and considering the polynomial system $\sysT[C, \bar{C}]$.

\section{The syntax of logic $\sysT$}\label{section:syntax}

Formally, a logic will consist of a countable set of atomic propositions $\Phi$ (generically denoted as $P,Q,R,\dots$) where terms are generated by taking a product, say $\muland$, with each other or with $\multrue$ along with a suitable use of parentheses. Rewrite rules for terms arise from inference, whose syntax we study next. We will use the intuitionistic sequent calculus and so write $\Gamma,\Delta, \dots$ for (not necessarily finite but possibly empty) sequences of terms. In the logic system $\system{T}$ we can freely commute terms in a sequent, formally called the ``exchange'' rule, and so we may take a sequent to be a multiset of terms. Some of our results will be sufficiently broad to include general sequents; we will denote as $\system{T}'$ the logic system with the same rules as $\system{T}$ except without the exchange rule. In particular when dealing with categorical semantics of our logics, we will find that $\system{T}'$ systems give rise to general monoidal categories while $\system{T}$ systems naturally lead to \emph{symmetric} monoidal categories. Throughout this section we will indicate where ordering, or a lack of it, is necessary in order to highlight the differences between $\sysT$ and $\sysT'$ systems.

In intuitionistic sequent calculus an inference is a relation between sequents and a term, written as $\Gamma \infers A$. We refer to terms in $\Gamma$ as \emph{antecedents} and $A$ as the \emph{consequent}. Intuitively, this is intended to represent that $A$ logically follows from the terms in $\Gamma$, however in a multiplicative logic we may not drop or introduce terms into $\Gamma$ and so a better interpretation may be that $A$ can be built out of the terms in $\Gamma$. Inferences are transformed according to rules that are meant to capture this intuition. For example, the exchange rule mentioned above states that terms in a sequent can be reordered arbitrarily. We can capture this by the rule
\begin{itemize}
    \item (Exchange) if $\Gamma, \Delta, \Theta, \Psi \infers A$ then $\Gamma, \Theta, \Delta, \Psi \infers A$.
\end{itemize}
Note that we have adhered to tradition by dropping cluttering brackets and braces in the antecedent and simply list the elements of the sequent, or in this case list its sub-sequents.

We will use proof trees throughout this work. These diagrams are a convenient notation for the application of deduction rules. For example the exchange rule above can be written as
\begin{prooftree}
    \rulename{Exchange}\AxiomC{$\Gamma, \Delta, \Theta, \Psi \infers A$}\UnaryInfC{$\Gamma, \Theta, \Delta, \Psi \infers A$.}
\end{prooftree}
In such a deduction, inferences above the line are called its \emph{assumptions} and the one below the line its \emph{conclusion}. A proof of an inference is a finite tree of rule applications that ends with the given inference as its conclusion.

We will denote inference in $\system{T}$ and $\system{T}'$ as $\Gamma \infers_\sysT A$ and $\Gamma \infers_{\system{T}'} A$ respectively, but drop the additional decoration when there is no danger of confusion or when the distinction between the systems is irrelevant. For example, the set of rules for $\sysT'$ is a subset of those for $\sysT$, and therefore any deduction in $\sysT'$ is also valid in $\sysT$. So, while defining rules for $\sysT'$, we may drop the unnecessary decorations as they hold in both logics. The syntax of deductions in $\sysT'$ is captured with the structural rules
\begin{itemize}
    \item (identity) for any atomic proposition $P$ we have $P \infers P$; and,
    \item (cut) for any term $B$, if $\Gamma \infers B$ and $\Delta, B, \Theta \infers A$ then $\Delta, \Gamma, \Theta \infers A$;
\end{itemize}
together with the logical rules for the nullary tautology and binary conjunction
\begin{itemize}
    \item (left-\mt) if $\Gamma, \Delta \infers A$ then $\Gamma, \multrue, \Delta \infers A$;
    \item (right-\mt) $\infers \multrue$;
    \item (left-\ma) if $\Delta, A, B, \Gamma \infers C$ then $\Delta, A \muland B, \Gamma \infers C$; and, 
    \item (right-\ma) if $\Gamma \infers A$ and $\Delta \infers B$ then $\Gamma,\Delta \infers A\muland B$.
\end{itemize}
We capture these as the deduction rules in \Cref{system-T':rules-table}.

\begin{table}[b]
    \centering
    \fbox{\begin{tabular}{c@{\qquad}c@{\qquad}c}
        \rulename{Id}\AxiomC{}\UnaryInfC{$P \infers P$}\DisplayProof &
        \rulename{Cut}\AxiomC{$\Gamma \infers B$}\AxiomC{$\Delta, B, \Theta \infers A$}\BinaryInfC{$\Delta, \Gamma, \Theta \infers A$}\DisplayProof &
        \rulename{\LI}\AxiomC{$\Gamma, \Delta \infers A$}\UnaryInfC{$\Gamma, \multrue, \Delta \infers A$}\DisplayProof\\\\ 
        \rulename{\LO}\AxiomC{$\Gamma, A, B, \Delta \infers C$}\UnaryInfC{$\Gamma, A\muland B, \Delta \infers C$}\DisplayProof &
        \rulename{\RO}\AxiomC{$\Gamma \infers A$}\AxiomC{$\Delta \infers B$}\BinaryInfC{$\Gamma,\Delta \infers A \muland B$}\DisplayProof &
        \rulename{\RI}\AxiomC{}\UnaryInfC{$\infers \multrue$}\DisplayProof\\\\
    \end{tabular}}
    \caption{Deduction rules of $\system{T}'$}
    \label{system-T':rules-table}
\end{table}

\begin{lemma}\label{systemT:order}
    In $\system{T}'$ (and hence also in $\system{T}$), if $A \infers B$ and $C \infers D$ then $A \ma C \infers B \ma D$.
\end{lemma}
\begin{proof}
    This is a straightforward use of the $\muland$-introduction rules:
    \begin{prooftree}
        \rulename{\RO}\AxiomC{$A \infers B$}\AxiomC{$C \infers D$}\BinaryInfC{$A,C \infers B \ma D$}
        \rulename{\LO}\UnaryInfC{$A \ma C \infers B \ma D$.}
    \end{prooftree}
\end{proof}

\begin{proposition}[Identity rule]\label{systemT:identity}
    In $\system{T}'$ (and hence also in $\system{T}$), we have $A \infers A$ for every term $A$.
\end{proposition}
\begin{proof}
    We work by structural induction on $A$. If $A = P$ is an atomic proposition then this inference is just the identity rule. If $A = \multrue$, then we prove $\multrue \infers \multrue$ as follows:
    \begin{prooftree}
        \rulename{R-$\mt$} \AxiomC{}\UnaryInfC{$\infers \mt$}
        \rulename{L-$\mt$} \UnaryInfC{$\mt \infers \mt.$}
    \end{prooftree}
    Now let $A = P \ma B$ for an atomic proposition $P$. Inductively, we may suppose that we have a proof of $B \infers B$. We also have $P \infers P$ from the identity rule and so by Lemma~\ref{systemT:order} $P\ma B \infers P\ma B$. The case of $A = B \ma P$ is identical.
    
    If $A = \mt \ma B$ then we proceed similarly: inductively we assume a proof of $B \infers B$, above we have given a proof of $\mt \infers \mt$, and so by Lemma~\ref{systemT:order}, we have a proof of $\mt \otimes B \infers \mt \ma B$. Again the case of $A = B \ma \mt$ is identical.
\end{proof}

\begin{lemma}[Left-$\mt$ elimination]\label{systemT:left-1-elimination}
    In $\system{T}'$ (and hence also in $\system{T}$), if $\Gamma, \mt, \Delta \infers A$ then $\Gamma, \Delta \infers A$.
\end{lemma}
\begin{proof}
    This follows immediately from the cut rule:
    \begin{prooftree}
        \rulename{R-$\mt$} \AxiomC{}\UnaryInfC{$\infers \mt$}
        \AxiomC{$\Gamma, \mt, \Delta \infers A$}
        \rulename{Cut} \BinaryInfC{$\Gamma, \Delta \infers A$.}
    \end{prooftree}
\end{proof}

\begin{lemma}[Right-$\mt$ elimination]\label{systemT:right-1-elimination}
    In $\system{T}'$ (and hence also in $\system{T}$), if $\Gamma \infers A \ma \mt \ma B$ then $\Gamma \infers A \ma B$. Similarly if $\Gamma \infers A \ma \mt$ then $\Gamma \infers A$ and if $\Gamma \infers \mt \ma B$ then $\Gamma \infers B$.
\end{lemma}
\begin{proof}
    The first of these implications also follows from the cut rule after constructing the the proper inference (in this case, $A \ma \mt \ma B \infers A \ma B$):
    \begin{prooftree}
        \rulename{Id}\AxiomC{}\UnaryInfC{$A \infers A$}
        \rulename{Id}\AxiomC{}\UnaryInfC{$B \infers B$}
        \rulename{\RO} \BinaryInfC{$A , B \infers A \ma B$}
        \rulename{L-$\mt$} \UnaryInfC{$A, \mt, B \infers A\ma B$}
        \rulename{(\LO) $\times 2$} \UnaryInfC{$A \ma \mt \ma B \infers A\ma B$}
        \AxiomC{$\Gamma \infers A \otimes \mt \otimes B$}
        \rulename{Cut} \BinaryInfC{$\Gamma \infers A\ma B$.}
    \end{prooftree}
    The two other proofs are identical.
\end{proof}

\begin{lemma}[Left-$\ma$ elimination]\label{lemma:system-T:left-tensor-elimination}
In $\system{T}'$ (and hence also in $\system{T}$), if $\Gamma, A\ma B, \Delta \infers C$ then $\Gamma, A, B, \Delta \infers C$.
\end{lemma}
\begin{proof}
    Like the proofs above, this follows from the cut rule after forming the proper inference
    \begin{prooftree}
        \rulename{Id}\AxiomC{}\UnaryInfC{$A \infers A$}
        \rulename{Id}\AxiomC{}\UnaryInfC{$B \infers B$}
        \rulename{R-$\ma$}\BinaryInfC{$A,B \infers A \ma B$}
        \AxiomC{$\Gamma, A\ma B, \Delta \infers C$}
        \rulename{Cut}\BinaryInfC{$\Gamma, A, B, \Delta \infers C$.}
    \end{prooftree}
\end{proof}

In system $\sysT$ we augment the structural rules of $\system{T}'$ with the \emph{exchange} rule, so that the structural rules for $\system{T}$ are
\begin{itemize}
    \item (identity) for any atomic proposition $P$ we have $P \infers P$;
	\item (exchange) if $\Gamma, \Delta, \Theta, \Psi \infers A$, then $\Gamma, \Theta, \Delta, \Psi \infers A$; and,
    \item (cut) for any term $B$, if $\Gamma \infers B$ and $\Delta, B \infers A$ then $\Gamma, \Delta \infers A$.
\end{itemize} 
The exchange rule explicitly states that the order of terms in the sequent on the left side can be arbitrarily reordered for a $\sysT$ inference. In particular, the cut rule presented for $\system{T}$ is the same as for $\system{T}'$ where the exchange rule has been used to simplify its form.

Similarly, the conjunction and tautology rules for $\sysT$ are those of $\sysT'$, except we have used the exchange rule to simplify some of their statements:
\begin{itemize}
	\item (left-$\multrue$) if $\Gamma \infers A$ then $\Gamma, \multrue \infers A$;
    \item (right-$\multrue$) $\infers \multrue$;
    \item (left-$\muland$) if $A, B, \Gamma \infers C$ then $A \muland B, \Gamma \infers C$; and 
    \item (right-$\muland$) if $\Gamma \infers A$ and $\Delta \infers B$ then $\Gamma,\Delta \infers A\muland B$.
\end{itemize}

We capture these as the deduction rules in \Cref{system-T:rules-table}. Notice that the left-$\muland$ introduction rule now simplifies to combine the terms in the order in which they appear---left to right. That is $A$ and $B$, the first two terms of the antecedent, combine to form the term $A \muland B$. However, as we can use the exchange rule to move any of the terms into the first two positions, and back, $\sysT$ is not limited in the expressing power of its inferences: anything inferred in $\sysT'$ can also be inferred in $\sysT$.

\begin{table}[b]
    \centering
    \fbox{\begin{tabular}{c@{\qquad}c@{\qquad}c@{\qquad}c}
        \rulename{Id}\AxiomC{}\UnaryInfC{$P \infers P$}\DisplayProof &
        \rulename{Cut}\AxiomC{$\Gamma \infers B$}\AxiomC{$\Delta, B \infers A$}\BinaryInfC{$\Delta, \Gamma \infers A$}\DisplayProof &
        \rulename{L-$\multrue$}\AxiomC{$\Gamma \infers A$}\UnaryInfC{$\Gamma, \multrue \infers A$}\DisplayProof & \rulename{R-$\multrue$}\AxiomC{}\UnaryInfC{$\infers \multrue$}\DisplayProof\\\\
        \rulename{\LO}\AxiomC{$A, B, \Gamma \infers C$}\UnaryInfC{$A\muland B, \Gamma \infers C$}\DisplayProof &
        \rulename{\RO}\AxiomC{$\Gamma \infers A$}\AxiomC{$\Delta \infers B$}\BinaryInfC{$\Gamma,\Delta \infers A \muland B$}\DisplayProof &
        \multicolumn{2}{c}{\rulename{Ex}\AxiomC{$\Gamma, \Delta, \Theta, \Psi \infers A$}\UnaryInfC{$\Gamma, \Theta, \Delta, \Psi \infers A$}\DisplayProof}
        \\\\
    \end{tabular}}
    \caption{Deduction rules of $\system{T}$}
    \label{system-T:rules-table}
\end{table}

We write $A \equiv B$ when $A \infers B$ and $B \infers A$. Note that when $A \equiv B$, the cut rule implies any appearance of $A$ in a proof can replaced with $B$, and vice versa.

\begin{proposition}\label{proposition:system-T:algebra}
    In $\sysT$ and $\sysT'$ we have
    \begin{enumerate}
        \item $(A \muland B) \muland C \equiv A \muland (B \muland C)$ and
        \item $A \muland \multrue \equiv A \equiv \multrue \muland A$.
    \end{enumerate}
    In $\sysT$ we additionally have 
    \begin{enumerate}
        \item[3.]  $A \muland B\equiv B\muland A$.
    \end{enumerate}
\end{proposition}
\begin{proof} 
The proof of one direction of (1) is as follows:
\begin{prooftree}
    \AxiomC{}\rulename{Id}\UnaryInfC{$A \infers A$}
    \AxiomC{}\rulename{Id}\UnaryInfC{$B \infers B$}
    \AxiomC{}\rulename{Id}\UnaryInfC{$C \infers C$}
    \rulename{\RO}\BinaryInfC{$B,C \infers B\muland C$}
    \rulename{\RO}\BinaryInfC{$A, B, C \infers A \muland (B \muland C)$}
    \rulename{\LO}\UnaryInfC{$A \muland B, C \infers A \muland (B \muland C)$}
    \rulename{\LO}\UnaryInfC{$(A \muland B) \muland C \infers A \muland (B \muland C)$.}
\end{prooftree}
The proof of the converse inference is essentially identical,
\begin{prooftree}
    \AxiomC{}\rulename{Id}\UnaryInfC{$A \infers A$}
    \AxiomC{}\rulename{Id}\UnaryInfC{$B \infers B$}
    \rulename{\RO}\BinaryInfC{$A,B \infers A\muland B$}\AxiomC{}
    \rulename{Id}\UnaryInfC{$C \infers C$}
    \rulename{\RO}\BinaryInfC{$A, B, C \infers (A \muland B) \muland C$}
    \rulename{\LO}\UnaryInfC{$A, B \muland C \infers (A \muland B) \muland C$}
    \rulename{\LO}\UnaryInfC{$A \muland (B \muland C) \infers (A \muland B) \muland C$.}
\end{prooftree}
The crux of the proof relies within the left-$\muland$ introduction rule---any two consecutive terms in the sequent may be combined. By combining them in a different order, we thus produce the ``associativity'' relation.

The proof of (2) is straight forward:
\begin{prooftree}
    \AxiomC{}\rulename{Id}\UnaryInfC{$A \infers A$}
    \rulename{L-$\multrue$}\UnaryInfC{$A, \multrue \infers A$}
    \rulename{\LO}\UnaryInfC{$A\muland \multrue \infers A$}
\end{prooftree}
and conversely
\begin{prooftree}
    \AxiomC{}\rulename{Id}\UnaryInfC{$A \infers A$}\AxiomC{}
    \rulename{R-$\multrue$}\UnaryInfC{$\infers \multrue$}
    \rulename{\RO}\BinaryInfC{$A \infers A\muland \multrue$.}
    \end{prooftree}
The proof for $\multrue \muland A$ is similar.

Finally, for (3) in $\system{T}$, we may use the exchange rule and hence give the following proof
\begin{prooftree}
    \AxiomC{}\rulename{Id}\UnaryInfC{$A \infers A$}\AxiomC{}\rulename{Id}\UnaryInfC{$B \infers B$}\rulename{R-$\muland$}\BinaryInfC{$A,B \infers A\muland B$}\rulename{Exchange}\UnaryInfC{$B, A \infers A \muland B$}\rulename{L-$\muland$}\UnaryInfC{$B\muland A \infers A \muland B$.}
\end{prooftree}
The converse inference is obtained by exchanging the roles of $A$ and $B$ above.
\end{proof}

Note that in this final part of the proof, the exchange rule was critical. So part (3) of the above proposition will not hold in general for $\system{T}'$ which fails to have this rule.

\section{Proof transformations}\label{section:transformations}

An essential component of our study is that of proof transformation. Rarely is a proof of an inference unique. But many proofs are considered equivalent as they simply involve rearranging the order of their rules. By swapping adjacent rules, or eliminating redundant rules, we define the following collection of proof transformations that will be used both in the cut elimination theorem (Theorem \ref{thm:cut_elim}) and later in Section \ref{section:category} to define the categorical semantics.

\subsection{Cut commutativity}

Consecutive cut rules lead to a canonical proof transformation, which states that the order of the two cuts is irrelevant to the proof.
\begin{equation}\label{equation:system-T:vertical-cut-commutivity}
\begin{array}{r@{\quad}l}
\text{\textsf{Cut $\circ$ Cut}:} &
\AxiomC{$\vdots$}\UnaryInfC{$\Gamma \infers B$}
\AxiomC{$\vdots$}\UnaryInfC{$\Delta, B, \Theta \infers C$}
\rulename{Cut}\BinaryInfC{$\Delta, \Gamma, \Theta \infers C$}
\AxiomC{$\vdots$}\UnaryInfC{$\Psi, C, \Xi \infers C$}
\rulename{Cut}\BinaryInfC{$\Psi, \Delta, \Gamma, \Theta, \Xi \infers C$} \DisplayProof\\\\ & \qquad \Rightarrow \qquad
\AxiomC{$\vdots$}\UnaryInfC{$\Gamma \infers B$}
\AxiomC{$\vdots$}\UnaryInfC{$\Delta, B, \Theta \infers C$}
\AxiomC{$\vdots$}\UnaryInfC{$\Psi, C, \Xi \infers C$}
\rulename{Cut}\BinaryInfC{$\Psi, \Delta, B, \Theta, \Xi \infers C$}
\rulename{Cut}\BinaryInfC{$\Psi, \Delta, \Gamma, \Theta, \Xi \infers C$.} \DisplayProof
\end{array}
\end{equation}
As opposed to consecutive cuts, parallel cuts also lead to a canonical proof transformation
\begin{equation}\label{equation:system-T:horizontal-cut-commutivity}
\begin{array}{r@{\quad}l}
\text{\textsf{Cut, Cut}:} &
\AxiomC{$\vdots$}\UnaryInfC{$\Gamma \infers A$}
\AxiomC{$\vdots$}\UnaryInfC{$\Delta \infers B$}
\AxiomC{$\vdots$}\UnaryInfC{$\Psi, A, \Xi, B, \Theta \infers C$}
\rulename{Cut}\BinaryInfC{$\Psi, A, \Xi, \Delta, \Theta \infers C$}
\rulename{Cut}\BinaryInfC{$\Psi, \Gamma, \Xi, \Delta, \Theta \infers C$} \DisplayProof\\\\ & \qquad \Rightarrow \qquad
\AxiomC{$\vdots$}\UnaryInfC{$\Delta \infers B$}
\AxiomC{$\vdots$}\UnaryInfC{$\Gamma \infers A$}
\AxiomC{$\vdots$}\UnaryInfC{$\Psi, A, \Xi, B, \Theta \infers C$}
\rulename{Cut}\BinaryInfC{$\Psi, \Gamma, \Xi, B, \Theta \infers C$}
\rulename{Cut}\BinaryInfC{$\Psi, \Gamma, \Xi, \Delta, \Theta \infers C$.} \DisplayProof
\end{array}
\end{equation}
A proof transformation we will rely upon heavily is cut decomposition: when cutting a composite term, we may transform the proof into one that cuts each of the sub-terms individually. Formally, 
\begin{equation}\label{equation:system-T:decomposing-cut}
\begin{array}{r@{\quad}l}
\text{\textsf{Cut $\ma$ Cut}:} &
    \AxiomC{$\vdots$}\UnaryInfC{$\Gamma \infers A$}
    \AxiomC{$\vdots$}\UnaryInfC{$\Delta \infers B$}
    \rulename{R-$\ma$}\BinaryInfC{$\Gamma, \Delta \infers A \ma B$}
    \AxiomC{$\vdots$}\UnaryInfC{$\Psi, A, B, \Theta \infers C$}
    \rulename{L-$\ma$}\UnaryInfC{$\Psi, A \ma B, \Theta \infers C$}
    \rulename{Cut}\BinaryInfC{$\Psi, \Gamma, \Delta, \Theta \infers C$.}\DisplayProof\\\\ & \qquad \Rightarrow \qquad
    \AxiomC{$\vdots$}\UnaryInfC{$\Gamma \infers A$}
    \AxiomC{$\vdots$}\UnaryInfC{$\Delta \infers B$}
    \AxiomC{$\vdots$}\UnaryInfC{$\Psi, A, B, \Theta \infers C$}
    \rulename{Cut}\BinaryInfC{$\Psi, A, \Delta, \Theta \infers C$}
    \rulename{Cut}\BinaryInfC{$\Psi, \Gamma, \Delta, \Theta \infers C$.} \DisplayProof
\end{array}
\end{equation}
Cut decomposition with the trivial decomposition (namely $\mt$) takes a special form, which we single out as the proof transformation
\begin{equation}\label{equation:system-T:removing-1-cut}
\begin{array}{r@{\quad}l}
\text{\textsf{$\mt$-Cut}:} &
    \rulename{R-$\mt$}\AxiomC{}\UnaryInfC{$\infers \mt$}
    \AxiomC{$\vdots$}\UnaryInfC{$\Psi, \Theta \infers C$}
    \rulename{L-$\mt$}\UnaryInfC{$\Psi, \mt, \Theta \infers C$}
    \rulename{Cut}\BinaryInfC{$\Psi, \Theta \infers C$.}\DisplayProof\\\\ & \qquad \Rightarrow \qquad
    \AxiomC{$\vdots$}\UnaryInfC{$\Psi, \Theta \infers C$.} \DisplayProof
\end{array}
\end{equation}

In a similar vein the cut rule can be commuted through a left-$\muland$ introduction. The two simpler cases occur when the left-$\muland$ introduction does not involve the term being cut. When this appears on the left the transformation is
\begin{equation}\label{equation:system-T:l-and-cut1}
\begin{array}{r@{\quad}l}
\text{\textsf{L-$\muland$/Cut(L)}:} &
\AxiomC{$\vdots$}\UnaryInfC{$\Gamma, C, D, \Delta \infers B$}
\rulename{L-$\muland$}\UnaryInfC{$\Gamma, C \muland D, \Delta \infers B$}
\AxiomC{$\vdots$}\UnaryInfC{$\Theta, B, \Psi \infers A$}
\rulename{Cut}\BinaryInfC{$\Theta, \Gamma, C \muland D, \Delta, \Psi \infers A$}\DisplayProof\\\\ & \qquad \Rightarrow \qquad
\AxiomC{$\vdots$}\UnaryInfC{$\Gamma, C, D, \Delta \infers B$}
\AxiomC{$\vdots$}\UnaryInfC{$\Theta, B, \Psi \infers A$}
\rulename{Cut}\BinaryInfC{$\Theta, \Gamma, C, D, \Delta, \Psi \infers A$}
\rulename{L-$\muland$}\UnaryInfC{$\Theta, \Gamma, C \muland D, \Delta, \Psi \infers A$.}\DisplayProof
\end{array}
\end{equation}
When this appears on the right one example of the transformation is
\begin{equation}\label{systemT:l-and-cut2}
\begin{array}{r@{\quad}l}
\text{\textsf{L-$\muland$/Cut(R)}:} &
\AxiomC{$\vdots$}\UnaryInfC{$\Gamma \infers B$}
\AxiomC{$\vdots$}\UnaryInfC{$\Delta, B, \Theta, C, D, \Psi \infers A$}
\rulename{L-$\muland$}\UnaryInfC{$\Delta, B, \Theta, C \muland D, \Psi \infers A$}
\rulename{Cut}\BinaryInfC{$\Delta, \Gamma, \Theta, C \muland D, \Psi \infers A$}\DisplayProof\\\\ & \qquad \Rightarrow \qquad
\AxiomC{$\vdots$}\UnaryInfC{$\Gamma \infers B$}
\AxiomC{$\vdots$}\UnaryInfC{$\Delta, B, \Theta, C, D, \Psi \infers A$}
\rulename{Cut}\BinaryInfC{$\Delta, \Gamma, \Theta, C, D, \Psi \infers A$}
\rulename{L-$\muland$}\UnaryInfC{$\Delta, \Gamma, \Theta, C \muland D, \Psi \infers A$.}\DisplayProof
\end{array}
\end{equation}
To be precise there is a similar transformation when the terms $C,D$ appear earlier in the sequent than the term to be cut $B$, however that figure is essentially identical. 

Likewise we can commute a cut rule with a right-$\ma$ introduction, as long as the combined term is not the one being cut. There is only one example of such a transformation
\begin{equation}\label{equation:system-T:r-and-cut}
\begin{array}{r@{\quad}l}
\text{\textsf{R-$\muland$/Cut}:} &
\AxiomC{$\vdots$}\UnaryInfC{$\Gamma \infers B$}
\AxiomC{$\vdots$}\UnaryInfC{$\Delta, B, \Theta \infers C$}
\AxiomC{$\vdots$}\UnaryInfC{$\Psi \infers D$}
\rulename{R-$\muland$}\BinaryInfC{$\Delta, B, \Theta, \Psi \infers C \muland D$}
\rulename{Cut}\BinaryInfC{$\Delta, \Gamma, \Theta, \Psi \infers C \muland D$} \DisplayProof\\\\ & \qquad \Rightarrow \qquad
\AxiomC{$\vdots$}\UnaryInfC{$\Gamma \infers B$}
\AxiomC{$\vdots$}\UnaryInfC{$\Delta, B, \Theta \infers C$}
\rulename{Cut}\BinaryInfC{$\Delta, \Gamma, \Theta \infers C$}
\AxiomC{$\vdots$}\UnaryInfC{$\Psi \infers D$}
\rulename{R-$\muland$}\BinaryInfC{$\Delta, \Gamma, \Theta, \Psi \infers C \muland D$.} \DisplayProof
\end{array}
\end{equation}

\subsection{Left and right identity transformations}

Suppose we have a proof of $\Gamma \infers A$. Then we construct a new proof by extending the given proof with the the identity and cut rules. Introducing the identity on the right gives the proof transformation
\begin{equation}\label{equation:system-T:right-identity-transformation}
\begin{array}{r@{\quad}c@{\qquad\Rightarrow\qquad}c}
\text{\textsf{R-Id}:} &
\AxiomC{$\vdots$}\UnaryInfC{$\Gamma \infers A$}\DisplayProof &
\AxiomC{$\vdots$}\UnaryInfC{$\Gamma \infers A$}
\rulename{Id}\AxiomC{}\UnaryInfC{$A \infers A$}
\rulename{Cut}\BinaryInfC{$\Gamma \infers A$.}\DisplayProof
\end{array}
\end{equation}
We may introduce an identity on the left for each term in $\Gamma$. Say $\Gamma = \Delta, B, \Theta$ gives a transformation 
\begin{equation}\label{equation:system-T:left-identity-transformation}
\begin{array}{r@{\quad}c@{\qquad\Rightarrow\qquad}c}
\text{\textsf{L-Id}:} &
\AxiomC{$\vdots$}\UnaryInfC{$\Delta,B,\Theta \infers A$}\DisplayProof &
\rulename{Id}\AxiomC{}\UnaryInfC{$B \infers B$}
\AxiomC{$\vdots$}\UnaryInfC{$\Delta,B,\Theta \infers A$}
\rulename{Cut}\BinaryInfC{$\Delta,B,\Theta \infers A$.}\DisplayProof
\end{array}
\end{equation}
The inverse of either of these transformations is an example of a cut-elimination: whenever the identity rule immediately precedes a cut deduction, both can be eliminated from the proof.

\subsection{Left-right introduction transformations}

Introduction rules for $\muland$ or $\multrue$ on the left or right naturally commute, which we capure in the following transformations.

\begin{equation}\label{equation:system-T:left-right-and-transformation}
\begin{array}{r@{\quad}c@{\qquad\Rightarrow\qquad}c}
    \text{\textsf{L-$\ma$}/\textsf{R-$\ma$}:} &
    \AxiomC{$\vdots$}\UnaryInfC{$\Gamma, A, B, \Delta \infers C$}
    \rulename{L-$\ma$}\UnaryInfC{$\Gamma, A \ma B, \Delta \infers C$}
    \AxiomC{$\vdots$}\UnaryInfC{$\Theta \infers D$}
    \rulename{R-$\ma$}\BinaryInfC{$\Gamma, A \ma B, \Delta, \Theta \infers C \ma D$}\DisplayProof &
    \AxiomC{$\vdots$}\UnaryInfC{$\Gamma, A, B, \Delta \infers C$}
    \AxiomC{$\vdots$}\UnaryInfC{$\Theta \infers D$}
    \rulename{R-$\ma$}\BinaryInfC{$\Gamma, A, B, \Delta, \Theta \infers C \ma D$}
    \rulename{L-$\ma$}\UnaryInfC{$\Gamma, A \ma B, \Delta, \Theta \infers C \ma D$.}\DisplayProof
\end{array}
\end{equation}

\begin{equation}\label{equation:system-T:left-1-right-and-transformation}
\begin{array}{r@{\quad}c@{\qquad\Rightarrow\qquad}c}
    \text{\textsf{L-$\mt$}/\textsf{R-$\ma$}:} &
    \AxiomC{$\vdots$}\UnaryInfC{$\Gamma, \Delta \infers B$}
    \rulename{L-$\mt$}\UnaryInfC{$\Gamma, \mt, \Delta \infers B$}
    \AxiomC{$\vdots$}\UnaryInfC{$\Theta \infers C$}
    \rulename{R-$\ma$}\BinaryInfC{$\Gamma, \mt, \Delta, \Theta \infers B \ma C$}\DisplayProof &
    \AxiomC{$\vdots$}\UnaryInfC{$\Gamma, \Delta \infers B$}
    \AxiomC{$\vdots$}\UnaryInfC{$\Theta \infers C$}
    \rulename{R-$\ma$}\BinaryInfC{$\Gamma, \Delta, \Theta \infers B \ma C$}
    \rulename{L-$\mt$}\UnaryInfC{$\Gamma, \mt, \Delta \Theta \infers B \ma C$.}\DisplayProof
\end{array}
\end{equation}

\section{The cut elimination theorem}\label{section:cut-elimination}

Gentzen's \emph{Hauptsatz} is cut elimination, which is to say that a proof in the system (the sequent calculus in his case) can be transformed into a proof of the stated inference that does not use the cut rule. Needless to say this can be quite involved. The  idea is to examine a general proof that ends in the cut rule
\begin{equation}\label{equation:System-T:cut}
    \AxiomC{$\Pi_1$}\noLine\UnaryInfC{$\vdots$}\UnaryInfC{$\Gamma \infers B$}
    \AxiomC{$\Pi_2$}\noLine\UnaryInfC{$\vdots$}\UnaryInfC{$\Delta, B \infers A$}
    \LeftLabel{$\Pi$:\quad}\rulename{Cut}\BinaryInfC{$\Gamma,\Delta \infers A$}
    \DisplayProof
\end{equation}
and attempt to commute the last step of either $\Pi_1$ or $\Pi_2$ with this cut rule. As such we can divide the proofs $\Pi_1$ and $\Pi_2$ into \emph{primary}, when this last step of the proof involves $A$ or $B$, or \emph{secondary}, when it involves $\Gamma$ or $\Delta$. The critical metric is the \emph{length} of a proof, which is the number of deductions used. Let us denote the length of a proof as $|\cdot|$; the length of proof (\ref{equation:System-T:cut}) is $|\Pi| = |\Pi_1| + |\Pi_2| + 1$.

A notable exception is when either $\Pi_1$ or $\Pi_2$ is the identity rule (either from an atomic proposition or via \cref{systemT:identity}), for which the cut can be eliminated directly using (\ref{equation:system-T:right-identity-transformation}, \ref{equation:system-T:left-identity-transformation}). That is, when $\Pi_1$ is the identity $\Pi$ is transformed
\begin{equation*}
\begin{array}{ccc}
    \AxiomC{}\UnaryInfC{$B \infers B$}
    \AxiomC{$\Pi_2$}\noLine\UnaryInfC{$\vdots$}\UnaryInfC{$\Delta, B \infers A$}
    \LeftLabel{$\Pi$:\quad}\rulename{Cut}\BinaryInfC{$\Delta, B \infers A$}\DisplayProof &
    \stackrel{\mathsf{L-Id}^{-1}}{\Longrightarrow} &
    \AxiomC{$\Pi_2$}\noLine\UnaryInfC{$\vdots$}\UnaryInfC{$\Delta, B \infers A$.}\DisplayProof
\end{array}
\end{equation*}
The new proof has length $|\Pi_2| = |\Pi|-2$. Similarly, when $\Pi_2$ is the identity
\begin{equation*}
\begin{array}{ccc}
    \AxiomC{$\Pi_1$}\noLine\UnaryInfC{$\vdots$}\UnaryInfC{$\Gamma \infers B$}
    \AxiomC{}\UnaryInfC{$B \infers B$}
    \LeftLabel{$\Pi$:\quad}\rulename{Cut}\BinaryInfC{$\Gamma \infers B$}\DisplayProof &
    \stackrel{\mathsf{R-Id}^{-1}}{\Longrightarrow} &
    \AxiomC{$\Pi_1$}\noLine\UnaryInfC{$\vdots$}\UnaryInfC{$\Gamma \infers B$}\DisplayProof
\end{array}
\end{equation*}
transforms $\Pi$ to a proof of length $|\Pi_1| = |\Pi| - 2$.

We begin with left-$\ma$ introduction. At this point, we will only consider the cases where it is secondary in $\Pi_1$ and $\Pi_2$, and reserve the case of being primary in $\Pi_2$ for later. When left-$\ma$ introduction is secondary in $\Pi_1$ we apply (\ref{equation:system-T:l-and-cut1}) to transform
\begin{equation*}
\begin{array}{cc@{\quad}c}
    \AxiomC{$\Pi_1'$}\noLine\UnaryInfC{$\vdots$}\UnaryInfC{$\Gamma, C, D\infers B$}\rulename{L-$\ma$}\UnaryInfC{$\Gamma,C\ma D \infers B$}
    \AxiomC{$\Pi_2$}\noLine\UnaryInfC{$\vdots$}\UnaryInfC{$\Delta, B \infers A$}
    \LeftLabel{$\Pi$:\quad}\rulename{Cut}\BinaryInfC{$\Gamma,C\ma D,\Delta \infers A$}\DisplayProof &
    \Longrightarrow &
    \AxiomC{$\Pi_1'$}\noLine\UnaryInfC{$\vdots$}\UnaryInfC{$\Gamma, C, D \infers B$}
    \AxiomC{$\Pi_2$}\noLine\UnaryInfC{$\vdots$}\UnaryInfC{$\Delta, B \infers A$}
    \LeftLabel{$\Pi'$:\quad}\rulename{Cut}\BinaryInfC{$\Gamma,\Delta, C, D \infers A$}\rulename{L-$\ma$}\UnaryInfC{$\Gamma, \Delta, C\ma D \infers A$.}\DisplayProof
\end{array}
\end{equation*}
Note that $|\Pi'| = |\Pi|$, however the cut rule is now part of a sub-proof of length $|\Pi|-1$. When left-$\ma$ introduction is secondary in $\Pi_2$ we apply (\ref{systemT:l-and-cut2}) to transform
\begin{equation*}
\begin{array}{cc@{\quad}c}
    \AxiomC{$\Pi_1$}\noLine\UnaryInfC{$\vdots$}\UnaryInfC{$\Gamma \infers B$}
    \AxiomC{$\Pi_2'$}\noLine\UnaryInfC{$\vdots$}\UnaryInfC{$\Delta, C, D, B \infers A$}\rulename{L-$\ma$}\UnaryInfC{$\Delta,C \ma D, B \infers A$}
    \LeftLabel{$\Pi$:\quad}\rulename{Cut}\BinaryInfC{$\Gamma,\Delta,C\ma D \infers A$}\DisplayProof &
    \Longrightarrow &
    \AxiomC{$\Pi_1$}\noLine\UnaryInfC{$\vdots$}\UnaryInfC{$\Gamma \infers B$}
    \AxiomC{$\Pi_2'$}\noLine\UnaryInfC{$\vdots$}\UnaryInfC{$\Delta, C, D, B \infers A$}
    \LeftLabel{$\Pi'$:\quad}\rulename{Cut}\BinaryInfC{$\Gamma,\Delta, C, D \infers A$}\rulename{L-$\ma$}\UnaryInfC{$\Gamma,\Delta,C\ma D \infers A$.}\DisplayProof
\end{array}
\end{equation*}
Similarly the cut rule is moved to a subproof of length $|\Pi|-1$.

Right-$\ma$ introduction involves the consequent of inferences, and so can only be primary in $\Pi_1$ and $\Pi_2$. We postpone the case when it is primary in $\Pi_1$ momentarily. When it is primary in $\Pi_2$ we apply (\ref{equation:system-T:r-and-cut}) to transform
\begin{equation*}
\begin{array}{cc@{\quad}c}
    \AxiomC{$\Pi_1$}\noLine\UnaryInfC{$\vdots$}\UnaryInfC{$\Gamma \infers B$}
    \AxiomC{$\Pi_{2,1}$}\noLine\UnaryInfC{$\vdots$}\UnaryInfC{$\Delta, B \infers C$}
    \AxiomC{$\Pi_{2,2}$}\noLine\UnaryInfC{$\vdots$}\UnaryInfC{$\Sigma \infers D$}
    \rulename{R-$\ma$}\BinaryInfC{$\Delta,\Sigma, B \infers C\ma D$}
    \LeftLabel{$\Pi$:\quad}\rulename{Cut}\BinaryInfC{$\Gamma,\Delta,\Sigma \infers C\ma D$}\DisplayProof &
    \Longrightarrow &
    \AxiomC{$\Pi_1$}\noLine\UnaryInfC{$\vdots$}\UnaryInfC{$\Gamma \infers B$}
    \AxiomC{$\Pi_{2,1}$}\noLine\UnaryInfC{$\vdots$}\UnaryInfC{$\Delta, B \infers C$}
    \LeftLabel{$\Pi'$:\quad}\rulename{Cut}\BinaryInfC{$\Gamma,\Delta \infers C$}
    \AxiomC{$\Pi_{2,2}$}\noLine\UnaryInfC{$\vdots$}\UnaryInfC{$\Sigma \infers D$}
    \rulename{R-$\ma$}\BinaryInfC{$\Delta,\Sigma, B \infers C\ma D$.}\DisplayProof
\end{array}
\end{equation*}
Just as before, $|\Pi'| = |\Pi|$ yet the cut rule is part of sub-proof of length $|\Pi|-1$.

Taking stock, we have examined every case for the last step of $\Pi_2$ except right-$\mt$ introduction and primary left-$\ma$ introduction. The first of these cannot happen: as right-$\mt$ introduction has no antecedents there are no terms to cut. Therefore we may assume that $\Pi_2$ ends in primary left-$\ma$ introduction and the proof has the form
\begin{prooftree}
    \AxiomC{$\Pi_1$}\noLine\UnaryInfC{$\vdots$}\UnaryInfC{$\Gamma \infers B\ma C$}
    \AxiomC{$\Pi_2'$}\noLine\UnaryInfC{$\vdots$}\UnaryInfC{$\Delta, B, C \infers A$}\rulename{L-$\ma$}\UnaryInfC{$\Delta, B\ma C \infers A$}
    \LeftLabel{$\Pi$:\quad}\rulename{Cut}\BinaryInfC{$\Gamma,\Delta \infers A$.}
\end{prooftree}
We note that this can only apply if the term being cut is composed of sub-terms via $\ma$, and so in particular cannot be $\mt$ or an atomic proposition. Similarly, we have considered every case for the last step of $\Pi_1$ except right-$\mt$ introduction and right-$\ma$ introduction. But this last step cannot be right-$\mt$ introduction as we have just seen that the term to be cut cannot be $\mt$. Therefore $\Pi_1$ ends in a right-$\ma$ introduction and so our proof is of the form 
\begin{prooftree}
    \AxiomC{$\Pi_{1,1}$}\noLine\UnaryInfC{$\vdots$}\UnaryInfC{$\Gamma \infers B$}
    \AxiomC{$\Pi_{1,2}$}\noLine\UnaryInfC{$\vdots$}\UnaryInfC{$\Sigma \infers C$}
    \rulename{R-$\ma$}\BinaryInfC{$\Gamma,\Sigma \infers B\ma C$}
    \AxiomC{$\Pi_2'$}\noLine\UnaryInfC{$\vdots$}\UnaryInfC{$\Delta, B, C \infers A$}\rulename{L-$\ma$}\UnaryInfC{$\Delta, B\ma C \infers A$}
    \LeftLabel{$\Pi$:\quad}\rulename{Cut}\BinaryInfC{$\Gamma,\Sigma,\Delta \infers A$.}
\end{prooftree}
Note $|\Pi| = |\Pi_{1,1}| + |\Pi_{1,2}| + |\Pi_2'| + 3$. We replace this proof with
\begin{prooftree}
    \AxiomC{$\Pi_{1,1}$}\noLine\UnaryInfC{$\vdots$}\UnaryInfC{$\Gamma \infers B$}
    \AxiomC{$\Pi_{1,2}$}\noLine\UnaryInfC{$\vdots$}\UnaryInfC{$\Sigma \infers C$}
    \AxiomC{$\Pi_2'$}\noLine\UnaryInfC{$\vdots$}\UnaryInfC{$\Delta, B, C \infers A$}
    \rulename{Cut}\BinaryInfC{$\Sigma,\Delta, B \infers A$}
    \LeftLabel{$\Pi'$:\quad}\rulename{Cut}\BinaryInfC{$\Gamma,\Sigma,\Delta \infers A$.}
\end{prooftree}
Note that while we have not commuted the final cut into an earlier proof, the new proof is shorter: $|\Pi'| = |\Pi_{1,1}| + |\Pi_{1,2}| + |\Pi_2'| + 2$.

\begin{theorem}[Cut elimination]
\label{thm:cut_elim}
    If a statement $\Gamma \infers A$ has a proof in System $\system{T}$ then it has a proof in System $\system{T}$ that does not use the cut rule.
\end{theorem}
\begin{proof}
    We induct on the length of a proof. If the proof of a statement is length $1$, then that statement must be either the identity, or right-$\mt$ introduction. Neither of these involve the cut rule. Inductively assume that any proof of length $n$ can be rewritten into a proof with length no more than $n$ in a way that does not involve the cut rule. Let $\Pi$ be a proof of $\Gamma \infers A$ of length $n+1$. Note that any cut rules appearing in sub-proofs of $\Pi$ may be can be eliminated inductively, so we may assume that only the last step of $\Pi$ may be the cut rule. If the last step of $\Pi$ is not the cut rule, then $\Pi$ is already a cut-free proof of $\Gamma \infers A$. If the last step of the proof of $\Gamma \infers A$ is the cut rule, then as we have shown above we either (i) eliminate the cut-rule obtaining a proof of length $\leq n$, (ii) commute the cut rule into a subproof where we may inductively eliminate it, or (iii) replace it with two cut rules in a new proof of overall length $n$ that again may be rewritten into a cut-free proof by the inductive hypothesis.
\end{proof}

\begin{corollary}\label{corollary:normal-form}
    Suppose $\Gamma \infers_\system{T} A$, then $\Gamma$ is finite. Moreover, by writing $\Gamma = B_1, \dots, B_k$, if we express
    \begin{itemize}
        \item $B_1 \ma \cdots \ma B_k = P_1 \ma \cdots \ma P_m \ma \mt^{\ma b}$ and
        \item $A = Q_1 \ma \cdots \ma Q_n \ma \mt^{\ma a}$
    \end{itemize}
    where $P_i,Q_j \in \Phi$ are atomic propositions then $m = n$ and there exists a permutation $\pi$ such that $Q_j = P_{\pi(i)}$. That is, modulo extra $\mt$ terms, $\Gamma$ and $A$ are formed from precisely the same multiset of atomic propositions. Conversely, if $\Gamma$ and $A$ are formed from the same multiset of atomic propositions, modulo extra $\mt$ terms, then $\Gamma \infers_\system{T} A$.
\end{corollary}
\begin{proof}
    Write a cut-free proof of $\Gamma \infers_\system{T} A$. Examining each of the deduction rules in this proof we see: (i) the same atomic proposition is added to antecedent and consequent (identity), (ii) a $\mt$ term is added to the antecedent or as a tensor factor to the consequent (left-$\mt$ or right-$\mt$ introduction), or the composition of atomic propositions (as factors in terms) in the antecedent and consequent are unchanged (left-$\ma$ or right-$\ma$ introduction or exchange). As there can be only finitely many rules in a proof, a finite number of atomic propositions and $\mt$ terms can be introduced and hence $\Gamma$ is finite. Moreover, atomic propositions are introduced to the antecedent and consequent in matching pairs, and therefore modulo $\mt$ terms the antecedent and consequent have the same factors up to reordering.
    
    The converse is straightforward: we construct a proof of $\Gamma \infers A$ by using the identity rule for each contained atomic proposition, use left-$\ma$ and right-$\ma$ introduction to assemble $A$ and the terms of $\Gamma$, and finally use left-$\mt$ and right-$\mt$ introduction to add any additional $\mt$ terms to the sequent and consequent.
\end{proof}

\section{Polynomial systems}\label{section:polynomials}

In this section we consider extensions to $\sysT$ (and also $\sysT'$) in terms of polynomial systems and provide an analogue of Lambek's lifting theorems \cite{lambek1974functional} for such extensions. From the resource theory perspective, these extensions will help to express resource conversions as well as the property of some resource to be freely available or freely disposable. Notice from left-$\mt$ and right-$\mt$ introduction rules in Table~\ref{system-T:rules-table} that the identity element $\multrue$ is the only one that can be arbitrarily added to the consequent or antecedent of any proof statement $\Gamma \infers A$. That is, $\multrue$ is a free resource. This indicates that adding rules similar to \LI\ and \RI\ for other terms can signify if they are free, as discussed below.

\begin{definition}
    Let $X$ be a term. Define $\system{T}[X]$ to be the logical system whose rules are all those of $\system{T}$ together with the right-$X$ introduction rule:
    \begin{itemize}
        \item (R-$X$) $\infers X$.
    \end{itemize}
\end{definition}

\begin{proposition}\label{proposition:polynomials:introduction-rules}
    In $\system{T}[X]$ we have
    \begin{enumerate}
        \item $\Gamma, X \infers A$ implies $\Gamma \infers A$, and
        \item $\Gamma \infers A$ implies $\Gamma \infers A\muland X$.
    \end{enumerate}
\end{proposition}
\begin{proof}
    The first result follows directly from the cut rule:
    \begin{prooftree}
        \rulename{R-$X$}\AxiomC{}\UnaryInfC{$\infers X$}\AxiomC{$\Gamma, X \infers A$}
        \rulename{Cut}\BinaryInfC{$\Gamma \infers A$.}
    \end{prooftree}
    The second one follows directly from right-$\muland$ introduction:
     \begin{prooftree}
        \AxiomC{$\Gamma \infers A$}\rulename{R-$X$}\AxiomC{}\UnaryInfC{$\infers X$}
        \rulename{Cut}\BinaryInfC{$\Gamma \infers A\muland X$.}
    \end{prooftree}
\end{proof}

The following result is a direct analogue of \cite[Proposition 1.5]{lambek1974functional}.

\begin{theorem}\label{systemT:lambek-lifting-theorem}
    There exists a proof of $\Gamma \infers_{\system{T}[X]} A$ if and only if there exists a proof of $\Gamma, X^m \infers_\system{T} A$ for some $m \geq 0$.
\end{theorem}
\begin{proof}
    Suppose that there exists a proof in $\system{T}$ of the form
    \begin{prooftree}
        \AxiomC{$\vdots$}\UnaryInfC{$\Gamma, X^m \infers A$.}
    \end{prooftree}
    Then as the rules of $\system{T}$ are a subset of those of $\system{T}[X]$ this is a proof in $\system{T}[X]$ as well. If $m = 0$ there is nothing to show, otherwise we extend this proof as follows:
     \begin{prooftree}
        \AxiomC{}\UnaryInfC{$\infers X$}
        \AxiomC{$\vdots$}\UnaryInfC{$\Gamma, X^{m-1}, X \infers A$}
        \rulename{Cut}\BinaryInfC{$\Gamma, X^{m-1} \infers A$.}
    \end{prooftree}
    We obtain a proof of $\Gamma \infers_{\system{T}[X]} A$ by induction.
    
    Conversely suppose we have a proof of $\Gamma \infers_{\system{T}[X]} A$; we focus on the locations in this proof where right-$X$ introduction was used. If the entire proof is
    \begin{prooftree}
        \rulename{R-$X$}\AxiomC{}\UnaryInfC{$\infers X$,}
    \end{prooftree}
    then we transform this into a proof in $\system{T}$ of the desired form (with $m=1$):
     \begin{prooftree}
        \rulename{Id}\AxiomC{}\UnaryInfC{$X\infers X$.}
    \end{prooftree}
    Otherwise this rule appears in the assumptions of another rule. If this other rule is the cut rule then this deduction has the form
    \begin{prooftree}
        \rulename{R-$X$}\AxiomC{}\UnaryInfC{$\infers X$}
        \AxiomC{$\vdots$}\UnaryInfC{$\Delta, X \infers B$}
        \rulename{Cut}\BinaryInfC{$\Delta \infers B$.}
    \end{prooftree}
    We transform this into the deduction
     \begin{prooftree}
        \rulename{Id}\AxiomC{}\UnaryInfC{$X \infers X$}
        \AxiomC{$\vdots$}\UnaryInfC{$\Delta, X \infers B$}
        \rulename{Cut}\BinaryInfC{$\Delta, X \infers B$.}
    \end{prooftree}
   If the other rule is left-$\multrue$ introduction then the deduction is of the form
    \begin{prooftree}
        \rulename{R-$X$}\AxiomC{}\UnaryInfC{$\infers X$}
        \rulename{L-$\multrue$}\UnaryInfC{$\multrue \infers X$.}
    \end{prooftree}
    Similarly, we transform this to the deduction
    \begin{prooftree}
        \rulename{Id}\AxiomC{}\UnaryInfC{$X \infers X$}
        \rulename{L-$\multrue$}\UnaryInfC{$\multrue, X \infers X$.}
    \end{prooftree}
    Finally, right-$\muland$ introduction is the only other rule that could have right-$X$ introduction as an assumption, which then has the form
    \begin{prooftree}
        \rulename{R-$X$}\AxiomC{}\UnaryInfC{$\infers X$}
        \AxiomC{$\vdots$}\UnaryInfC{$\Delta \infers B$}
        \rulename{R-$\muland$}\BinaryInfC{$\Delta \infers X \muland B$.}
    \end{prooftree}
    We transform this to
     \begin{prooftree}
        \rulename{Id}\AxiomC{}\UnaryInfC{$X \infers X$}
        \AxiomC{$\vdots$}\UnaryInfC{$\Delta \infers B$}
        \rulename{R-$\muland$}\BinaryInfC{$\Delta, X \infers X \muland B$.}
    \end{prooftree}
   These transformations eliminate all instances of the right-$X$ introduction rule at the cost of introducing one $X$ into the antecedent for each instance. Hence the proof of $\Gamma \infers_{\system{T}[X]} A$ is transformed into a proof of $\Gamma, X^m \infers_{\system{T}} A$ where $m$ is the number of uses of right-$X$ introduction.
\end{proof}

We extend to more terms inductively: $\system{T}[X_1, \dots, X_n] = \system{T}[X_1, \dots, X_{n-1}][X_n]$. We leave it to the reader to adapt the proof of the theorem above for the following result.

\begin{corollary}
    There exists a proof of $\Gamma \infers_{\system{T}[X_1,\dots,X_n]} A$ if and only if there exists a proof of $\Gamma, X_1^{m_1}, \dots, X_n^{m_n} \infers_\system{T} A$ for some $m_j \geq 0$ and $j=1, \dots, n$.
\end{corollary}

\begin{definition}
    Let $X$ be a term. Define $\system{T}[\bar{X}]$ to be the logical system whose rules are all those of $\system{T}$ together with
    \begin{itemize}
        \item (L-$X$) $X \infers \multrue$.
    \end{itemize}
\end{definition}

\begin{proposition}\label{proposition:polynomials:dual-introduction-rules}
    In $\system{T}[\bar{X}]$ we have
    \begin{enumerate}
        \item $\Gamma \infers A$ implies $\Gamma, X \infers A$, and
        \item $\Gamma \infers A\otimes X$ implies $\Gamma \infers A$.
    \end{enumerate}
\end{proposition}
\begin{proof}
    In \Cref{proposition:system-T:algebra} we proved $A \otimes \multrue \infers A$ for any term $A$. Together with cut and right-$\muland$ introduction this gives the first result:
    \begin{prooftree}
        \AxiomC{$\Gamma \infers A$}
        \rulename{L-$X$}\AxiomC{}\UnaryInfC{$X \infers \multrue$}
        \rulename{R-$\muland$}\BinaryInfC{$\Gamma, X \infers A \otimes \multrue$}
        \rulename{Prop. \ref{proposition:system-T:algebra}}\AxiomC{}\UnaryInfC{$A\otimes \multrue \infers A$}
        \rulename{Cut}\BinaryInfC{$\Gamma, X \infers A$.}
    \end{prooftree}
    
    The second results follows from the first:
    \begin{prooftree}
        \AxiomC{$\Gamma \infers A\otimes X$}
        \AxiomC{}\UnaryInfC{$A \infers A$}
        \rulename{(Part 1)}\UnaryInfC{$A, X \infers A$}
        \rulename{L-$\otimes$}\UnaryInfC{$A \muland X \infers A$}
        \rulename{Cut}\BinaryInfC{$\Gamma \infers A$.}
    \end{prooftree}
\end{proof}

\begin{proposition}\label{systemT:disposable-lifting-theorem}
    There exists a proof of $\Gamma \infers_{\system{T}[\bar{X}]} A$ if and only if there exists a proof of $\Gamma \infers_\system{T} A\otimes X^{\ma m}$ for some $m \geq 0$.
\end{proposition}
\begin{proof}
    The proof is essentially the same as that of \Cref{systemT:lambek-lifting-theorem} and so is left to the reader.
\end{proof}

The two constructions above characterize when a term is freely available and freely disposable respectively. However in standard resource theories, a class of free transformations is specified and the free states are those that are produced from $I$ using free transformations. As we wish to follow a strict Howard-Curry-Lambek correspondence, a resource conversion from $A$ to $B$ should be represented as the inference $A \infers B$ and this being a free transformation implies the rule $\AxiomC{}\UnaryInfC{$A \infers B$}\DisplayProof$ should be added to our logic. We write this system as $\sysT[A \infers B]$. Note that if $A$ and $B$ have a common atomic proposition, $A = A' \otimes P$ and $B = B' \otimes P$, then we have the proof
$$\AxiomC{$A' \infers B'$}
\rulename{Id}\AxiomC{}\UnaryInfC{$P \infers P$}
\rulename{R-$\ma$}\BinaryInfC{$A',P \infers B'\ma P$}
\rulename{L-$\ma$}\UnaryInfC{$A \infers B$.}\DisplayProof$$
Thus we may instead use the system $\sysT[A' \infers B']$, as we can prove $A \infers B$ therein. That is, when we consider a resource conversion from $A$ to $B$ and form $\sysT[A \infers B]$, we may without loss of generality assume that $A$ and $B$ have no common atomic propositions.

We claim that we can construct $\sysT[A \infers B]$ with polynomials. There are two dual approaches, it is unclear if one has any advantages over the other. The key step is to formally introduce a term $-A$, which behaves in many ways like the negation of $A$. We then work in the polynomial system $\system{T}[-A \ma B, \overline{A \ma -A}]$. In this system we prove
$$\rulename{Id}\AxiomC{}\UnaryInfC{$A \infers A$}
\rulename{Prop. \ref{proposition:polynomials:introduction-rules}}\UnaryInfC{$A \infers A \ma (-A \ma B)$}
\rulename{Prop. \ref{proposition:system-T:algebra}}\UnaryInfC{$A\infers (A\ma -A)\ma B$}
\rulename{L-$(A\ma -A)$}\AxiomC{}\UnaryInfC{$A \ma -A \infers \mt$}
\rulename{Id}\AxiomC{}\UnaryInfC{$B \infers B$}
\rulename{\RO; \LO}\BinaryInfC{$(A \ma -A) \ma B \infers \mt \ma B$}
\rulename{Cut}\BinaryInfC{$A \infers B$.}
\DisplayProof$$
Alternatively we introduce the term $-B$ and work in the polynomial system $\system{T}[\overline{A \ma -B}, -B \ma B]$. We prove
$$\rulename{R-$(-B \ma B)$}\AxiomC{}\UnaryInfC{$\infers -B \ma B$}
\rulename{Id}\AxiomC{}\UnaryInfC{$B \infers B$}
\rulename{Prop. \ref{proposition:polynomials:dual-introduction-rules}}\UnaryInfC{$(A \ma -B), B \infers B$}
\rulename{L-$\ma$; Prop. \ref{proposition:system-T:algebra}}\UnaryInfC{$A \ma (-B \ma B) \infers B$}
\rulename{Lem. \ref{lemma:system-T:left-tensor-elimination}}\UnaryInfC{$A, (-B \ma B) \infers B$}
\rulename{Cut}\BinaryInfC{$A \infers B$.}
\DisplayProof$$

As we can prove $A \infers B$ in either of these systems, any inference we can prove in $\sysT[A \infers B]$ we can prove in both $\system{T}[-A \ma B, \overline{A \ma -A}]$ and $\system{T}[\overline{A \ma -B}, -B \ma B]$. Our claim is that the converse is also true, after a fashion. Namely, $-A$ is not a term in $\sysT[A \infers B]$ and hence no inferences will involve it. Hence we focus on inferences provable in $\system{T}[-A \ma B, \overline{A \ma -A}]$ or $\system{T}[\overline{A \ma -B}, -B \ma B]$ that do not involve $-A$.

\begin{proposition}
    Suppose $C \infers_{\system{T}[-A \ma B, \overline{A \ma -A}]} D$ where neither $C$ nor $D$ contains $-A$ as a subterm. Then $C \infers_{\system{T}[A \infers B]} D$.
\end{proposition}
\begin{proof}
    Using the lifting theorems above we have $C, (-A \ma B)^m \infers_\sysT D \otimes (A \ma -A)^{\ma n}$, or equivalently
    $$C, (-A)^m, B^m \infers_\sysT D\otimes A^{\ma n} \otimes (-A)^{\ma n}.$$
    From Corollary \ref{corollary:normal-form}, the atomic propositions constituting both sides of this inference must be in one-to-one correspondence. Since neither $C$ nor $D$ contains $-A$ as a subterm, we must then have $n = m$. Moreover, we have assumed $A$ and $B$ contain no common atomic propositions and hence $C = C' \ma A^{\ma n}$ and $D = D' \ma B^{\ma n}$ where $C'$ and $D'$ are composed of the same multiset of atomic propositions. That is $C' \infers_\sysT D'$. We extend the proof of $C' \infers D'$ in $\sysT[A\infers B]$ as follows:
    $$\AxiomC{$\vdots$}\UnaryInfC{$C' \infers D'$}
    \AxiomC{}\UnaryInfC{$A \infers B$}
    \rulename{R-$\ma$; L-$\ma$}\BinaryInfC{$C' \ma A \infers D' \ma B$}
    \noLine\UnaryInfC{$\vdots$}
    \UnaryInfC{$C' \ma A^{\ma (n-1)} \infers D' \ma B^{\ma (n-1)}$}
    \AxiomC{}\UnaryInfC{$A \infers B$}
    \rulename{R-$\ma$; L-$\ma$}\BinaryInfC{$C \infers D$}
    \DisplayProof$$
\end{proof}

An identical proof works for $\system{T}[\overline{A \ma -B}, -B \ma B]$ where we instead focus on inferences that do not involve the term $-B$. The claim is stated below for completeness.

\begin{proposition}
    Suppose $C \infers_{\system{T}[\overline{A \ma -B}, -B \ma B]} D$ where neither $C$ nor $D$ contains $-B$ as a subterm. Then $C \infers_{\system{T}[A \infers B]} D$.
\end{proposition}

\begin{example}\label{example:coherence}
Here we form a simple type system for the resource theory of coherence for qubit systems: the free states are diagonal density operators (classical states) with respect to some fixed basis, and the resource content of a state is its $L^1$-distance from this convex set. We bend the rules somewhat to allow an uncountably infinite family of atomic propositions:
$$\Phi = \{ Q(\lambda) \::\: \lambda \in [0,1] \}.$$
We identify $Q(0)$ as the type of classical states in the given basis. That is, in a programming language realizing this type theory we would write
$$\begin{pmatrix} p & 0 \\ 0 & 1-p\end{pmatrix} \mathsf{:: Q(0)}.$$
There are a variety of choices for what we might consider for the incoherent operations (that is, quantum channels that preserve $Q(0)$). The most general set of such operations, so-called MIO, is somewhat intractable; the most commonly studied set of incoherent operations, IO, is characterized as those quantum channels that admit a set of Kraus operators $\{K_n\}$ such that if $\rho \in Q(0)$ then for all $n$ we have $\frac{K_n\rho K_n^\dagger}{\mathrm{tr}(K_n\rho K_n^\dagger)} \in Q(0)$\footnote{as opposed to MIO where merely $\sum_n K_n\rho K_n^\dagger \in Q(0)$.} \cite{baumgratz2014quantifying}. See \cite{streltsov2017colloquium} for a detailed discussion of numerous other choices.

Clearly the totally depolarizing channel is in IO, thus there a free transformation of states of any type to classical states. Logically this implies that the inference $Q(\lambda) \infers Q(0)$ is a theorem. However states in $Q(0)$ are free, and so we have $Q(0) \infers \mt$ and hence using the cut rule $Q(\lambda)\infers \mt$ in general.

In \cite{baumgratz2014quantifying} it is shown that there is a maximally coherent state from which any other state can be prepared with a channel in IO. The maximally coherent qubit state is of type $Q(1)$, and consequently our resource theory logic also contains the inferences $Q(1) \infers Q(\lambda)$.

The conditions for which one can transform one qubit state into another by an incoherent map are known but nontrivial \cite{gour2017quantum}. In fact, the requirement that a channel reduce the $L^1$ distance from the classical states is precisely half of the characterization given as \cite[equations (3,4)]{streltsov2017structure}. At the level of types we cannot guarantee further free transformations, and so our logic is the system $$\sysT\left[ \{Q(0)\} \cup \{\overline{Q(\lambda)} \::\: \lambda \in [0,1]\} \cup \{Q(1) \infers Q(\lambda) \::\: \lambda \in [0,1)\}\right].$$
\end{example}

\begin{example}\label{example:locc}
Extending Example \ref{example:introduction} in a different direction, we can build a toy model of qubit entanglement. Here we have actors Alice and Bob who have access to local quantum computation on their respective systems together with arbitrary classical communication (LOCC), see for instance \cite{chitambar2014everything}. Hence we pose atomic propositions $Q_A$ for the type of a qubit held by Alice, $Q_B$ for a qubit held by Bob, and $C$ for a classical bit which given arbitrary communication may be considered shared between Alice and Bob. For illustrative purposes we simply add an additional atomic proposition $E$ for the type of a maximally entangled state shared by Alice and Bob. Needless to say, this toy model is insufficient to capture the full power of LOCC. Nonetheless we can capture the types of simple quantum protocols; for example the quantum teleportation circuit has type $E\ma Q_A \infers Q_B$. This is obtained by composing Alice's measurement of type $E\ma Q_A \infers C \ma C\ma Q_B$ with Bob's correction of type $C\ma C\ma Q_B \infers Q_B$.

As classical information may be freely created or destroyed $\infers C$ and $C \infers \mt$ are theorems of this logic. Similarly local measurements are free in LOCC and thus $Q_A \infers C$ and $Q_B \infers C$ are theorems, and similar to the coherence case above so also are $Q_A\infers \mt$ and $Q_B \infers \mt$. In dealing with $E$ there are some options. We could form a logic where $E \infers Q_A\ma Q_B$ is a theorem; in such a system entanglement is disposable in the technical sense given above. Namely any protocol requiring $Q_A$ and $Q_B$ can be run with $E$ instead, and any protocol producing $E$ and be transformed into one producing $Q_A\ma Q_B$. This logic is the system
$$\sysT[C, \overline{C}, \overline{Q_A}, \overline{Q_B}, E \infers Q_A\ma Q_B].$$
On the other hand, we may take the system where merely $E\infers C\ma Q_B$ and $E\infers Q_A\ma C$ are theorems. In this system, entanglement can only be disposed of by either Alice or Bob measuring his or her half of the entangled pair. This system is strong enough to contain the teleportation protocol as a theorem, but is strictly weaker than the previous system\footnote{We challenge the reader to prove $E \infers Q_A\ma Q_B$ is not a theorem of this latter system.}. 
\end{example}

\section{Algebraic Semantics}\label{section:algebra}

In the previous sections we have treated resources as terms in a logic with resource conversions as inferences in this logic. In \cite{fritz2015resource} abstract resources theories are cast the algebraic language of monoids, with resource conversion described by an order in the monoid. Here we prove soundness and completeness of the algebraic semantics of $\sysT$ through ordered monoids.

\begin{definition}
    A \emph{monoid} is a set of elements $\mathbb{M}$ together with a binary operation $\monoid$ and element $e \in \mathbb{M}$ satisfying
    \begin{itemize}
        \item $(x \monoid y) \monoid z = x \monoid (y \monoid z)$, and
        \item $x \monoid e = x = e \monoid x$.
    \end{itemize}
    A monoid is said to be \emph{commutative} if it additionally satisfies
    \begin{itemize}
        \item $x \monoid y = y \monoid x$.
    \end{itemize}
    A preorder $\leq$ (i.e. a reflexive and transitive order) on a monoid is a \emph{monoid order} if it satisfies
    \begin{itemize}
        \item $r \leq s$ and $x \leq y$ implies $r\cdot x \leq s \cdot y$.
    \end{itemize}
    An \emph{ordered monoid} (respectively, \emph{ordered commutative monoid}) is a monoid (respectively, commutative monoid) together with a given monoid order.
\end{definition}

Formally, an algebraic model $\mathcal{M}$ of $\system{T}$ consists of an ordered commutative monoid $(\mathbb{M},\monoid,e,\leq)$ together with a valuation on atomic propositions $v:\Phi \to \mathbb{M}$. Intuitively, the valuation encodes the resources associated to each proposition in the given model. For general terms we use a \emph{forcing} relation $\forces$ defined inductively as follows:
\begin{enumerate}
    \item if $P\in \Phi$ is atomic then $m \forces P$ means $m \leq v(P)$;
    \item $m \forces A \ma B$ if and only if there exists $n_1,n_2\in \mathbb{M}$ with $m = n_1\monoid n_2$ and $n_1 \forces A$ and $n_2 \forces B$; and 
    \item $m \forces \mt$ if and only if $m = e$.
\end{enumerate}
Extending the valuation to terms homomorphically, $v(A \ma B) = v(A) \monoid v(B)$ with $v(\mt) = e$, we then have $m \forces A$ implies that $m \leq v(A)$ for all terms. However the converse is not generally true as the Joyal-like rule (2) is stronger than merely stating that $m \leq v(A)\cdot v(B)$. In the language of \cite{coecke2013resource}, monoids where these are equivalent are called ``noninteracting'' although other names appear in the literature. Nonetheless we still have $v(A) \forces A$ for any term, since we can decompose $A = P_1\muland \cdots \muland P_k$ where all the $P_j$ are atomic and $v(P_j) \forces P_j$.

In the language of a resource-based type theory, the relation $m \forces A$ indicates that $A$ contains enough resources for the context described by $m$ to exist. In other words, $m \in \mathbb{M}$ gives resource requirements for some system (at this level of abstraction, a system can be considered synonymous with its resource requirements) and $m \forces A$ means that instances of type $A$ can provide the needed resources to instantiate $m$. Entailment of types over a requirement $m\in \mathbb{M}$ has the usual logical formulation:
\begin{itemize}
    \item $\mathcal{M}: A \manifests_m B$  means either $m \forces B$ or $m \not\forces A$ (that is, $m \forces A$ implies $m \forces B$).
\end{itemize}
This extends to sequents in the natural way:
\begin{itemize}
    \item $\mathcal{M}:\ \manifests_m B$ if and only if $m \forces B$ or $m \not= e$, and
    \item $\mathcal{M}: A_1, \dots, A_k \manifests_m B$ if and only if whenever there exists $m = m_1\cdots m_k$ with $m_j \forces A_j$ (for each $j=1, \dots, k)$ then $m \forces B$.
\end{itemize}
Semantic entailment in the model is defined as
$$\mathcal{M}: \Gamma \manifests B \text{ means } \mathcal{M}: \Gamma \manifests_m B \text{ for all } m\in \mathbb{M},$$
and general semantic entailment is
$$\Gamma \manifests B \text{ means } \mathcal{M}: \Gamma \manifests B \text{ for all models } \mathcal{M}.$$

\begin{theorem}[Strong soundness]\label{theorem:algebra:soundness}
    If $\Gamma \infers_{\system{T}} A$ then $\Gamma \manifests A$.
\end{theorem}
\begin{proof}
    Note that by \cref{corollary:normal-form} we must have $\Gamma$ is finite and hence semantic entailment is well defined. Now fix an algebraic model $\mathcal{M} = (\mathbb{M},\monoid, e, \leq, v)$. We prove soundness by algebraic induction on a proof of $\Gamma \infers A$. First if our proof is just
    $$\rulename{Id}\AxiomC{}\UnaryInfC{$P \infers P$}\DisplayProof$$
    then the result is tautological: either $m \forces P$ or $m \not\forces P$. Similarly if the proof is
    $$\rulename{R-$\mt$}\AxiomC{}\UnaryInfC{$\infers \mt$}\DisplayProof$$
    then again this is a tautology: either $m \not= e$ or $m=e$ (and hence $m \forces \mt$).
    
    Now suppose the last step of the proof is the cut rule
    $$\AxiomC{$\Delta \infers B$}\AxiomC{$\Theta, B \infers A$}
    \rulename{Cut}\BinaryInfC{$\Theta, \Delta \infers A$.}\DisplayProof$$
    We must show $\mathcal{M}: \Delta, \Theta \manifests_m A$ and so assume $\Delta = D_1, \dots, D_k$ and $\Theta = E_1, \dots, E_\ell$ where $m_j \forces D_j$ and $m_j' \forces E_j$ and let $m = m_1\cdots m_k \monoid m_1' \cdots m_\ell'$. Inductively we have $\mathcal{M}: \Delta \manifests_m B$ and so can conclude that $m_1\cdots m_k \forces B$. Again, inductively, $\mathcal{M}: \Theta, B \manifests_m A$ and $m_1\cdots m_k \forces B$ imply that $m = m_1\cdots m_k \monoid m_1' \cdots m_\ell' \forces A$ as desired.
    
    If the last step of the proof is left-$\mt$ introduction
    $$\rulename{L-$\mt$}\AxiomC{$\Delta\infers A$}\UnaryInfC{$\Delta,\mt \infers A$,}\DisplayProof$$
    we need to show $\mathcal{M}: \Delta, \mt \manifests_m A$ and so assume $m = m_0\monoid m_1 \cdots m_k$ with $m_0 \forces \mt$ and $m_j \forces D_j$ (where again we write $\Delta = D_1, \dots, D_k$).
    Then we have $m_0 = e$ and by definition, $m_1 \cdots m_k \forces A$ inductively and therefore $m = m_1\cdots m_k \forces A$ as desired. 
    
    Likewise if the last step of the proof is left-$\ma$ introduction
    $$\rulename{L-$\ma$}\AxiomC{$B,C,\Delta\infers A$}\UnaryInfC{$B\otimes C,\Delta \infers A$,}\DisplayProof$$
    then we assume $m = m_0\monoid m_1 \cdots m_k$ with $m_0 \forces B\otimes C$ and $m_j \forces D_j$. Now by definition, we have $m_0 = n_1 \monoid n_2$ with $n_1 \forces B$ and $n_2 \forces C$. Inductively, we have $m = n_1 \monoid n_2 \monoid m_1 \cdots m_k \forces A$ as needed.
    
    Finally if the last step of the proof is right-$\ma$ introduction
    $$\AxiomC{$\Delta \infers B$}\AxiomC{$\Theta \infers C$}
    \rulename{R-$\ma$}\BinaryInfC{$\Delta, \Theta \infers B\otimes C$,}\DisplayProof$$
    then as before we take $\Delta = D_1, \dots, D_k$ and $\Theta = E_1, \dots, E_\ell$ with $m = m_1\cdots m_k \monoid m_1' \cdots m_\ell'$ and $m_j \forces D_j$ and $m_j' \forces E_j$. Then inductively $m_1\cdots m_k \forces B$ and $m_1' \cdots m_\ell' \forces C$ and so by definition $m = m_1\cdots m_k \monoid m_1' \cdots m_\ell' \forces B \ma C$.
\end{proof}

\begin{theorem}[Strong completeness]\label{theorem:algebra:completeness}
    If $\Gamma \manifests A$ then $\Gamma \infers_{\system{T}} A$.
\end{theorem}
\begin{proof}
    Assume $\Gamma \manifests A$, and write $\Gamma = B_1, \dots, B_k$. Let $(\mathbb{M},\cdot,e)$ be the free monoid generated by the atomic propositions $\Phi$ identifying $e = \mt$. Define $C \leq D$ to mean $C = D$, and take the valuation to be trivial $v(C) = C$. Let $\mathcal{M}$ be the model with this data. Then in this model, we claim $m \forces A$ if and only if $m = A$. We have already seen $A = v(A) \forces A$; for the converse, we have also seen that $m \forces A$ implies $m \leq v(A)$, which here means $m = v(A)$. In particular, by taking $m = \prod_{j=1}^k v(B_j)$ and noting $\mathcal{M}: \Gamma \manifests_m A$ we conclude $m \forces A$. But this implies $B_1\ma \cdots \ma B_k = A$. Expanding both sides into generators from $\Phi$, we can then apply \cref{corollary:normal-form} to conclude $\Gamma \infers A$.
\end{proof}

\begin{example}\label{example:coherence-algebra}
Continuing example \ref{example:coherence}, the atomic propositions of the logic are $\{Q(\lambda)\}_{\lambda \in [0,1]}$. Then the free monoid appearing in Theorem \ref{theorem:algebra:completeness} is isomorphic to the coproduct $\coprod_{\lambda\in [0,1]} \mathbb{N}$ where $\mathbb{N}$ is the natural numbers under addition, and $e$ as the all zeros tuple. Equivalently this moniod can be realized as the collection of  finite multisets of elements from $[0,1]$ with union as the operation and $e$ as the empty multiset.

While a full treatment of algebraic models for polynomial systems over $\sysT$ is beyond the scope of this paper, we indicate that every model of $\sysT$ is also a model of a polynomial system. However we cannot expected soundness and completeness for such models as polynomial systems have additional rules that every model has to adhere to. For instance, In the example of coherence, $Q(0)$ is free and so we enforce that every model's valuation have $v(Q(0)) = e$. Similarly, each $Q(\lambda)$ is freely disposable and hence we also require that every model's valuation have $v(Q(\lambda)) \leq e$. 
\end{example}

\section{Categorical Semantics}\label{section:category}

In the algebraic semantics of the previous section the focus was on terms; existence of an inference was relegated to an ordered relation. Categorical semantics centers on the proofs of inferences. In fact, the formalism of categories and logics are so similar that soundness and completeness (typically unified under the name ``coherence'') are immediate, once one establishes that the logic forms a free category of the type under consideration. Here, this will be a monoidal category for $\system{T}'$ and a symmetric monoidal category for $\system{T}$.

A category $\catC$ is a class of objects, together with a sets $\hom(A,B)$ for each pair of objects $A,B$ and pairings
$$\_ \circ \_ : \hom(B,C) \times \hom(A,B) \to \hom(A,C)$$
that satisfies certain rules:
\begin{enumerate}
    \item $(f\circ g)\circ h = f\circ (g\circ h)$ whenever these are defined, and
    \item for each object $A$ there is $\id_A\in\hom(A,A)$ which satisfies
    $$\id_B\circ f = f \circ \id_A = f \text{ whenever } f \in \hom(A,B).$$
\end{enumerate}

A category is monoidal if it has a monoid-like operation on its objects \cite{maclane1963natural}. The main issue stemming from this informality is that algebraic operations are defined on sets, and the objects of a category typically are not a set. Consequently, operations are promoted to functors, and the rules they must satisfy are described as natural transformations. For example, a binary operation $\monoidal$ is promoted to a covariant bifunctor $\_ \monoidal \_: \category{C} \times \category{C} \to \category{C}$. This means that to any pair of objects $(A,B)$ of $\category{C} \times \category{C}$ the functor produces a new object $A \monoidal B$, and to every pair of homomorphisms $(f,g)\in \hom_\category{C}(A,C) \times \hom_\category{C}(B,D)$ the functor produces a homomorphism $f\monoidal g \in \hom_\category{C}(A\otimes B, C\otimes D)$. 

\begin{definition}
    Let \category{C} be a category. A \emph{monoidal structure} on \category{C} consists of a covariant bifunctor $\_\monoidal\_:\category{C}\times\category{C} \to \category{C}$, an object $I$, and three natural isomorphisms 
    \begin{itemize}
        \item $\alpha: (\_ \monoidal \_) \monoidal \_ \cong \_\monoidal (\_ \times \_)$,
        \item $\lambda: I \monoidal \_ \cong \_$, and
        \item $\rho: \_ \monoidal I \cong \_$
    \end{itemize}
    such that the two diagrams commute: the triangle rule
    \begin{equation}\label{eqn:MoC:triangle}
        \begin{tikzpicture}[baseline=(current bounding box.center)]
            \draw (0,1.5) node (n0) {$(A\monoidal I) \monoidal B$};
            \draw (4,1.5) node (n1) {$A\monoidal (I \monoidal B)$};
            \draw (2,0) node (n2) {$A \monoidal B$,};
            \draw (2,1.7) node (a0) {\scriptsize $\alpha_{A,T,B}$};
            \draw (0.3,0.6) node (a1) {\scriptsize $\rho_A\monoidal\id_B$};
            \draw (3.8,0.6) node (a2) {\scriptsize $\id_A\monoidal\lambda_B$};
            \draw[->] (n0) -- (n1);
            \draw[->] (n1) -- (n2);
            \draw[->] (n0) -- (n2);
        \end{tikzpicture}
    \end{equation}
    and the pentagon rule
    \begin{equation}\label{eqn:MoC:pentagon}
        \begin{tikzpicture}[baseline=(current bounding box.center)]
            \draw (1.5,3) node (n0) {$(A \monoidal (B \monoidal C)) \monoidal D$};
            \draw (0,1.5) node (n1) {$((A \monoidal B) \monoidal C) \monoidal D$};
            \draw (4,0) node (n2) {$(A \monoidal B) \monoidal (C \monoidal D)$.};
            \draw (6.5,3) node (n3) {$A \monoidal ((B \monoidal C) \monoidal D)$};
            \draw (8,1.5) node (n4) {$A \monoidal (B \monoidal (C \monoidal D))$};
            \draw (4,3.2) node (a0) {\scriptsize $\alpha_{A,B\monoidal C,D}$};
            \draw (-0.3,2.35) node (a1) {\scriptsize $\alpha_{A,B,C}\monoidal \id_D$};
            \draw (8.3,2.35) node (a2) {\scriptsize $\id_A\monoidal\alpha_{B,C,D}$};
            \draw (1,0.6) node (a3) {\scriptsize $\alpha_{A\monoidal B,C,D}$};
            \draw (7,0.6) node (a4) {\scriptsize $\alpha_{A,B,C\monoidal D}$};
            \draw[->] (n0) -- (n1);
            \draw[->] (n1) -- (n2);
            \draw[->] (n0) -- (n3);
            \draw[->] (n3) -- (n4);
            \draw[->] (n4) -- (n2);
        \end{tikzpicture}
    \end{equation}
    A category with monoidal structure we call a \emph{monoidal category}.
\end{definition}

\begin{definition}
    Let \category{C} be a category. A \emph{symmetric monoidal structure} on \category{C} is a monoidal structure together with a family of isomorphisms $\sigma_{A,B}\in\hom(A\monoidal B, B\monoidal A)$ for which the following diagrams commute:
    \begin{equation}\label{eqn:SMC-symmetry}
        \begin{tikzcd}
            A\monoidal I \arrow[rr, "\sigma_{A,T}"'] \arrow[rd, "\rho_A"'] & & I \monoidal A \arrow[ld, "\lambda_A"]\\
            & A,
        \end{tikzcd}
    \end{equation}
    \begin{equation}\label{eqn:SMC-inverse}
        \begin{tikzcd}
            A\monoidal B \arrow[rr, "\id_{A\monoidal B}"'] \arrow[rd, "\sigma_{A,B}"'] & & A \monoidal B \\
            & B\monoidal A, \arrow[ru, "\sigma_{B,A}"']
        \end{tikzcd}
    \end{equation}
 and the hexagon rule
    \begin{equation}\label{eqn:SMC-hexagon}
        \begin{tikzcd}
            & (A\monoidal B) \monoidal C \arrow[rr, "\sigma_{A,B} \monoidal \id_C"] \arrow[ld, "\alpha_{A,B,C}"'] & & (B \monoidal A) \monoidal C \arrow[rd, "\alpha_{B,A,C}"] & \\
            A \monoidal (B \monoidal C) \arrow[rd, "\sigma_{A, B \monoidal C}"'] & & & & B \monoidal (A\monoidal C) \arrow[ld, "\id_B \monoidal \sigma_{A,C}"]\\
            & (B \monoidal C) \monoidal A \arrow[rr, "\alpha_{B,C,A}"'] & & B \monoidal (C \monoidal A). &
        \end{tikzcd}
    \end{equation}
\end{definition}

\begin{proposition}
    Let $\Phi$ be a countable set of atomic propositions and $\sysT'$ be the logic formed from it. Then the following data forms a category $\category{C}$:
    \begin{itemize}
        \item each term $A$ defines an object of $\category{C}$;
        \item each pair of terms $A,B$ define $\hom_\category{C}(A,B)$ as the equivalence class of proofs of $A \infers B$ modulo any equivalence relation containing the proof transformations (\ref{equation:system-T:vertical-cut-commutivity}, \ref{equation:system-T:right-identity-transformation}, \ref{equation:system-T:left-identity-transformation}) and their inverses;
        \item for each term $A$ define $\id_A = \frames{\rulename{Id}\AxiomC{}\UnaryInfC{$A\infers A$}\DisplayProof} \in \hom_\category{C}(A, A)$; and,
        \item for $\frames{\Pi_1} \in \hom_\category{C}(A, B)$ and $\frames{\Pi_2} \in \hom_\category{C}(B, C)$ define $\frames{\Pi_2}\circ\frames{\Pi_1} = \frames{\Pi_3}$ where
        $$\AxiomC{$\Pi_1$}\noLine\UnaryInfC{$\vdots$}\UnaryInfC{$A \infers B$}
        \AxiomC{$\Pi_2$}\noLine\UnaryInfC{$\vdots$}\UnaryInfC{$B \infers C$}
        \LeftLabel{$\Pi_3 := \quad$}\rulename{Cut}\BinaryInfC{$A \infers C$.}\DisplayProof$$
    \end{itemize}
\end{proposition}
\begin{proof}
    To prove associativity of composition consider 
    $$\begin{array}{rcl}\Pi_4 &:=& \AxiomC{$\Pi_1$}\noLine\UnaryInfC{$\vdots$}\UnaryInfC{$A \infers B$}
    \AxiomC{$\Pi_2$}\noLine\UnaryInfC{$\vdots$}\UnaryInfC{$B \infers C$}
    \rulename{Cut}\BinaryInfC{$A \infers C$}
    \AxiomC{$\Pi_3$}\noLine\UnaryInfC{$\vdots$}\UnaryInfC{$C \infers D$}
    \rulename{Cut}\BinaryInfC{$A \infers D$}\DisplayProof
    \end{array}$$
    and
    $$\begin{array}{rcl}\Pi'_4 &:=& \AxiomC{$\Pi_1$}\noLine\UnaryInfC{$\vdots$}\UnaryInfC{$A \infers B$}
    \AxiomC{$\Pi_2$}\noLine\UnaryInfC{$\vdots$}\UnaryInfC{$B \infers C$}
    \AxiomC{$\Pi_3$}\noLine\UnaryInfC{$\vdots$}\UnaryInfC{$C \infers D$}
    \rulename{Cut}\BinaryInfC{$B \infers D$}
    \rulename{Cut}\BinaryInfC{$A \infers D$.}\DisplayProof
    \end{array}$$
    Then by cut commutativity (\ref{equation:system-T:vertical-cut-commutivity}) we have
    $$(\frames{\Pi_1}\circ\frames{\Pi_2})\circ \frames{\Pi_3} = \frames{\Pi_4} =  \frames{\Pi'_4} = \frames{\Pi_1}\circ(\frames{\Pi_2}\circ \frames{\Pi_3}).$$
    
    The identity transformations (\ref{equation:system-T:right-identity-transformation}) and (\ref{equation:system-T:left-identity-transformation}) show $\id$ is the the identity morphism. Namely,
    $$\begin{array}{rcccccl}
    \id_B \circ \frames{\Pi}
    &=&
    \frames{%
    \AxiomC{$\Pi$}\noLine\UnaryInfC{$\vdots$}\UnaryInfC{$A \infers B$}
    \rulename{Id}\AxiomC{}\UnaryInfC{$B \infers B$}
    \rulename{Cut}\BinaryInfC{$A \infers B$}
    \DisplayProof}
    &\stackrel{(\ref{equation:system-T:right-identity-transformation})}{=}& 
    \frames{%
    \AxiomC{$\Pi$}\noLine\UnaryInfC{$\vdots$}\UnaryInfC{$A \infers B$}
    \DisplayProof}
    &=& \frames{\Pi}
    \end{array}$$
    and similarly on the right.
\end{proof}

\begin{lemma}
    Let $\category{C}$ be as above. Define $A \monoidal B = A \ma B$ and for proofs $\Pi_1$ of $A \infers B$ and $\Pi_2$ of $C \infers D$ define
    $$\frames{\Pi_1} \monoidal \frames{\Pi_2} = \frames{%
    \AxiomC{$\Pi_1$}\noLine\UnaryInfC{$\vdots$}\rulename{}\UnaryInfC{$A \infers B$}
	\AxiomC{$\Pi_2$}\noLine\UnaryInfC{$\vdots$}\rulename{}\UnaryInfC{$C \infers D$}
	\rulename{\RO}\BinaryInfC{$A, C \infers B \muland D$}
	\rulename{\LO}\UnaryInfC{$A \muland C \infers B \muland D$.}\DisplayProof}$$
	If proofs related by transformations (\ref{equation:system-T:decomposing-cut}, \ref{equation:system-T:l-and-cut1}, \ref{equation:system-T:r-and-cut}) or their inverses are equivalent then $\_ \monoidal \_ : \category{C} \times \category{C} \to \category{C}$ is a bifunctor.
\end{lemma}
\begin{proof}
     We need only show the functorial relation $$(\frames{\Pi_1} \monoidal \frames{\Pi_2}) \circ (\frames{\Pi_3} \monoidal \frames{\Pi_4}) = (\frames{\Pi_1} \circ \frames{\Pi_3}) \monoidal (\frames{\Pi_2} \circ \frames{\Pi_4})$$
     Starting with the left side we have
     $$\AxiomC{$\Pi_1$}\noLine\UnaryInfC{$\vdots$}\UnaryInfC{$A \infers C$}
     \AxiomC{$\Pi_2$}\noLine\UnaryInfC{$\vdots$}\UnaryInfC{$B \infers D$}
     \rulename{R-$\ma$}\BinaryInfC{$A, B \infers C \ma D$}
     \rulename{L-$\ma$}\UnaryInfC{$A \ma B \infers C \ma D$}
     \AxiomC{$\Pi_3$}\noLine\UnaryInfC{$\vdots$}\UnaryInfC{$C \infers E$}
     \AxiomC{$\Pi_4$}\noLine\UnaryInfC{$\vdots$}\UnaryInfC{$D \infers F$}
     \rulename{R-$\ma$}\BinaryInfC{$C, D \infers E \ma F$}
     \rulename{L-$\ma$}\UnaryInfC{$C \ma D \infers E \ma F$}
     \rulename{Cut}\BinaryInfC{$A \ma B \infers E \ma F$.}
     \DisplayProof$$
     Commuting the cut through the left-$\ma$ introduction using ($\ref{equation:system-T:l-and-cut1}$) this proof is transformed to
     $$\AxiomC{$\Pi_1$}\noLine\UnaryInfC{$\vdots$}\UnaryInfC{$A \infers C$}
     \AxiomC{$\Pi_2$}\noLine\UnaryInfC{$\vdots$}\UnaryInfC{$B \infers D$}
     \rulename{R-$\ma$}\BinaryInfC{$A, B \infers C \ma D$}
     \AxiomC{$\Pi_3$}\noLine\UnaryInfC{$\vdots$}\UnaryInfC{$C \infers E$}
     \AxiomC{$\Pi_4$}\noLine\UnaryInfC{$\vdots$}\UnaryInfC{$D \infers F$}
     \rulename{R-$\ma$}\BinaryInfC{$C, D \infers E \ma F$}
     \rulename{L-$\ma$}\UnaryInfC{$C \ma D \infers E \ma F$}
     \rulename{Cut}\BinaryInfC{$A, B \infers E \ma F$}
     \rulename{L-$\ma$}\UnaryInfC{$A \ma B \infers E \ma F$.}
     \DisplayProof$$
     Now we apply the cut decomposition rule (\ref{equation:system-T:decomposing-cut}) to obtain
     $$\AxiomC{$\Pi_1$}\noLine\UnaryInfC{$\vdots$}\UnaryInfC{$A \infers C$}
     \AxiomC{$\Pi_2$}\noLine\UnaryInfC{$\vdots$}\UnaryInfC{$B \infers D$}
     \AxiomC{$\Pi_3$}\noLine\UnaryInfC{$\vdots$}\UnaryInfC{$C \infers E$}
     \AxiomC{$\Pi_4$}\noLine\UnaryInfC{$\vdots$}\UnaryInfC{$D \infers F$}
     \rulename{R-$\ma$}\BinaryInfC{$C, D \infers E \ma F$}
     \rulename{Cut}\BinaryInfC{$C, B \infers E \ma F$}
     \rulename{Cut}\BinaryInfC{$A, B \infers E \ma F$}
     \rulename{L-$\ma$}\UnaryInfC{$A \ma B \infers E \ma F$.}
     \DisplayProof$$
     We now move the top cut across the right-$\ma$ introduction rule with (\ref{equation:system-T:r-and-cut}) giving proof
     $$\AxiomC{$\Pi_1$}\noLine\UnaryInfC{$\vdots$}\UnaryInfC{$A \infers C$}
     \AxiomC{$\Pi_3$}\noLine\UnaryInfC{$\vdots$}\UnaryInfC{$C \infers E$}
     \AxiomC{$\Pi_2$}\noLine\UnaryInfC{$\vdots$}\UnaryInfC{$B \infers D$}
     \AxiomC{$\Pi_4$}\noLine\UnaryInfC{$\vdots$}\UnaryInfC{$D \infers F$}
     \rulename{Cut}\BinaryInfC{$B \infers F$}
     \rulename{R-$\ma$}\BinaryInfC{$C, B \infers E \ma F$}
     \rulename{Cut}\BinaryInfC{$A, B \infers E \ma F$}
     \rulename{L-$\ma$}\UnaryInfC{$A \ma B \infers E \ma F$.}
     \DisplayProof$$
     Finally moving the lower cut across the right-$\ma$ introduction rule with (\ref{equation:system-T:r-and-cut}) transforms our proof to
     $$\AxiomC{$\Pi_1$}\noLine\UnaryInfC{$\vdots$}\UnaryInfC{$A \infers C$}
     \AxiomC{$\Pi_3$}\noLine\UnaryInfC{$\vdots$}\UnaryInfC{$C \infers E$}
     \rulename{Cut}\BinaryInfC{$A \infers E$}
     \AxiomC{$\Pi_2$}\noLine\UnaryInfC{$\vdots$}\UnaryInfC{$B \infers D$}
     \AxiomC{$\Pi_4$}\noLine\UnaryInfC{$\vdots$}\UnaryInfC{$D \infers F$}
     \rulename{Cut}\BinaryInfC{$B \infers F$}
     \rulename{R-$\ma$}\BinaryInfC{$A, B \infers E \ma F$}
     \rulename{L-$\ma$}\UnaryInfC{$A \ma B \infers E \ma F$.}
     \DisplayProof$$
     But this is precise the proof associated to the right side of the functorial relation above.
\end{proof}

Additionally, we define $I = \mt$. The monoidal structure is given through three natural isomorphisms $\alpha, \lambda, \rho$. The first one encodes associativity of $\monoidal$, while the last two ensure that $I$ behaves as a left and right identity respectively. Define $\lambda_A \in \hom(I \monoidal A, A)$ and $\rho_A \in \hom(A \monoidal I, A)$ as in Proposition \ref{proposition:system-T:algebra}:
\begin{equation}
\label{eq:rho-lambda}
    \begin{array}{ccc}
        \lambda_A = \frames{\AxiomC{}\rulename{Id}\UnaryInfC{$A \infers A$}
        \rulename{L-$\multrue$}\UnaryInfC{$\mt, A \infers A$}
        \rulename{\LO}\UnaryInfC{$\mt \ma A \infers A$}\DisplayProof}
        & \text{ and } &
        \rho_A = \frames{\AxiomC{}\rulename{Id}\UnaryInfC{$A \infers A$}
        \rulename{L-$\multrue$}\UnaryInfC{$A, \multrue \infers A$}
        \rulename{\LO}\UnaryInfC{$A\muland \multrue \infers A$}\DisplayProof}.
    \end{array}
\end{equation}

\begin{lemma}
    If proofs related by transformations (\ref{equation:system-T:decomposing-cut}, \ref{equation:system-T:removing-1-cut}, \ref{equation:system-T:left-identity-transformation}) or their inverses are equivalent then $\lambda, \rho$ are natural isomorphisms.
\end{lemma}
\begin{proof}
    The inverse morphisms are given by the proofs of $A \infers \mt \ma A$ and $A \infers A \ma \mt$ given in Proposition \ref{proposition:system-T:algebra}. Namely, composing $\rho_A$ with its putative inverse is the image of the proof
\begin{equation*}
    \rulename{Id}\AxiomC{}\UnaryInfC{$A \infers A$}
    \rulename{R-$\mt$}\AxiomC{}\UnaryInfC{$\infers \mt$}
    \rulename{R-$\ma$}\BinaryInfC{$A \infers A \ma \mt$}
    \AxiomC{}\rulename{Id}\UnaryInfC{$A \infers A$}
    \rulename{L-$\multrue$}\UnaryInfC{$A, \multrue \infers A$}
    \rulename{\LO}\UnaryInfC{$A\muland \multrue \infers A$}
    \rulename{Cut}\BinaryInfC{$A \infers A$.}\DisplayProof
\end{equation*}
Clearly this is a proof of $A \infers A$, however it must be equivalent to the identity. To show this we apply a sequence of reducing proof transformations as follows.
\begin{align*}
    \frames{\AxiomC{}\UnaryInfC{$A \infers A$}
    \AxiomC{}\UnaryInfC{$\infers \mt$}
    \BinaryInfC{$A \infers A \ma \mt$}
    \AxiomC{}\UnaryInfC{$A \infers A$}
    \UnaryInfC{$A, \multrue \infers A$}
    \UnaryInfC{$A\muland \multrue \infers A$}
    \BinaryInfC{$A \infers A$}\DisplayProof}
    &\stackrel{(\ref{equation:system-T:decomposing-cut})}{=}
    \frames{\AxiomC{}\UnaryInfC{$A \infers A$}
    \AxiomC{}\UnaryInfC{$\infers \mt$}
    \AxiomC{}\UnaryInfC{$A \infers A$}\UnaryInfC{$A, \mt \infers A$}
    \BinaryInfC{$A\infers A$}\BinaryInfC{$A\infers A$}
    \DisplayProof}\\
    &\stackrel{(\ref{equation:system-T:left-identity-transformation})}{=}
    \frames{\AxiomC{}\UnaryInfC{$\infers \mt$}
    \AxiomC{}\UnaryInfC{$A \infers A$}\UnaryInfC{$A, \mt \infers A$}
    \BinaryInfC{$A\infers A$}\DisplayProof}\\
    &\stackrel{(\ref{equation:system-T:removing-1-cut})}{=}
    \frames{\AxiomC{}\UnaryInfC{$A \infers A$}\DisplayProof}.
\end{align*}
An identical argument works for $\lambda_A$.

To be a natural transformation we require the functorial properties
\begin{equation*}
\begin{tikzcd}
 \frames{A} \monoidal I \arrow[r, "\rho_A"] \arrow[d, "\frames{\Pi} \monoidal \id_{I}"'] & \frames{A} \arrow[d, "\frames{\Pi}"] \\
 \frames{B} \monoidal I \arrow[r, "\rho_B"'] & \frames{B}
\end{tikzcd}
\quad \text{ and } \quad
\begin{tikzcd}
 I \monoidal \frames{A} \arrow[r, "\lambda_A"] \arrow[d, "\id_{I} \monoidal \frames{\Pi}"'] & \frames{A} \arrow[d, "\frames{\Pi}"] \\
 I \monoidal \frames{B} \arrow[r, "\lambda_B"'] & \frames{B}
\end{tikzcd}
\end{equation*}
In other words, we need to show that $\frames{\Pi} \circ \rho_A \equiv \rho_b \circ \frames{\Pi} \monoidal I$. By definition $\frames{\Pi} \circ \rho_A$ is given by 
\begin{equation*}
    \AxiomC{}\rulename{Id}\UnaryInfC{$A \infers A$}
    \rulename{L-$\multrue$}\UnaryInfC{$A, \multrue \infers A$}
    \rulename{\LO}\UnaryInfC{$A\muland \multrue \infers A$}
    \AxiomC{$\Pi$}\noLine\UnaryInfC{$\vdots$}\UnaryInfC{$A \infers B$}
    \rulename{Cut}\BinaryInfC{$A \ma \mt \infers B$.}\DisplayProof
\end{equation*}
With a slight reshuffle of inferences to remove the additional $\mt$, explicitly using the definition of $\rho_A$ and eliminating cuts, the commutative diagram can be shown to hold. A similar proof also holds for $\lambda_A$ and $\lambda_B$.  
\end{proof}

Next we define $\alpha_{A,B,C} \in \hom((A \monoidal B) \monoidal C, A \monoidal (B \monoidal C))$ as in the proof of Proposition \ref{proposition:system-T:algebra}:
\begin{equation*}
    \alpha_{A,B,C} = \frames{%
    \AxiomC{}\rulename{Id}\UnaryInfC{$A \infers A$}
    \AxiomC{}\rulename{Id}\UnaryInfC{$B \infers B$}
    \AxiomC{}\rulename{Id}\UnaryInfC{$C \infers C$}
    \rulename{\RO}\BinaryInfC{$B,C \infers B\muland C$}
    \rulename{\RO}\BinaryInfC{$A, B, C \infers A \muland (B \muland C)$}
    \rulename{\LO}\UnaryInfC{$A \muland B, C \infers A \muland (B \muland C)$}
    \rulename{\LO}\UnaryInfC{$(A \muland B) \muland C \infers A \muland (B \muland C)$}
    \DisplayProof}.
\end{equation*}
That this is an isomorphism also follows from this proposition.

\begin{equation}
\begin{tikzcd}
 (A \muland B) \muland C \arrow[rr, "\alpha_{A,B,C}"] \arrow[dd, "(f \muland g) \muland h"'] & & A \muland (B \muland C) \arrow[dd, "f \muland (g \muland h)"] \\ 
 & & \\
 (D \muland E) \muland F \arrow[rr, "\alpha_{D,E,F}"'] & & D \muland (E \muland F)
\end{tikzcd}
\end{equation}

Finally to form a monoidal category, $\alpha, \lambda, \rho$ must satisfy the triangle and pentagon rules. We prove these in turn as follows.

\begin{lemma}[Triangle Rule]
Let logic system $\sysT'$ give a category $\catC$ as discussed above and suppose that proofs related by transformation rules (\ref{equation:system-T:decomposing-cut}, \ref{equation:system-T:l-and-cut1}, \ref{equation:system-T:left-identity-transformation}, \ref{equation:system-T:left-right-and-transformation}, \ref{equation:system-T:left-1-right-and-transformation}) or their inverses are equivalent. Then, the triangle coherence relation is satisfied:
\begin{equation*}
    \begin{tikzpicture}[baseline=(current bounding box.center)]
            \draw (0,1.5) node (n0) {$(A\monoidal I) \monoidal B$};
            \draw (4,1.5) node (n1) {$A\monoidal (I \monoidal B)$};
            \draw (2,0) node (n2) {$A \monoidal B$};
            \draw (2,1.7) node (a0) {\scriptsize $\alpha_{A,I,B}$};
            \draw (0.3,0.6) node (a1) {\scriptsize $\rho_A\monoidal\id_B$};
            \draw (3.8,0.6) node (a2) {\scriptsize $\id_A\monoidal\lambda_B$};
            \draw[->] (n0) -- (n1);
            \draw[->] (n1) -- (n2);
            \draw[->] (n0) -- (n2);
    \end{tikzpicture}
\end{equation*}
\end{lemma}
\begin{proof}
We must show:
\begin{equation}
\label{eq:triangle}
	\alpha_{A, \multrue, B} \circ (\id_A \muland \lambda_B) = \rho_A \muland \id_B.
\end{equation}
First we compute $\alpha_{A, \multrue, B} \circ (\id_A \muland \lambda_B)$ as follows:
\begin{eqnarray*}
\scpf
\begin{tikzcd}
(A \monoidal I) \monoidal B \arrow[dd, "\alpha_{A, \mt, B}"'] \\
\\
A \monoidal (I \monoidal B)
\end{tikzcd}
\;\; &=& \;\;
\frames{\AxiomC{}\rulename{Id}\UnaryInfC{$A \infers A$}
\AxiomC{}\rulename{Id}\UnaryInfC{$\mt \infers \mt$}
\AxiomC{}\rulename{Id}\UnaryInfC{$B \infers B$}
\rulename{\RO}\BinaryInfC{$\mt,B \infers \mt \muland B$}
\rulename{\RO}\BinaryInfC{$A, \mt, B \infers A \muland (\mt \muland B)$}
\rulename{\LO}\UnaryInfC{$A \muland \mt, B \infers A \muland (\mt \muland B)$}
\rulename{\LO}\UnaryInfC{$(A \muland \mt) \muland B \infers A \muland (\mt \muland B)$}\DisplayProof}\\
\begin{tikzcd}
A \monoidal (I \monoidal B) \arrow[dd, "\id_A \ma \lambda_B"'] \\
\\
A \monoidal B
\end{tikzcd}
\;\; &=& \;\;\frames{
\AxiomC{}\rulename{Id}\UnaryInfC{$A \infers A$}
\AxiomC{}\rulename{Id}\UnaryInfC{$B \infers B$}
\rulename{\LI}\UnaryInfC{$\mt, B \infers B$}
\rulename{\LO}\UnaryInfC{$\mt \ma B \infers B$}
\rulename{\RO}\BinaryInfC{$A, \mt \ma B \infers A \ma B$}
\rulename{\LO}\UnaryInfC{$A \ma (\mt \ma B) \infers A \ma B$}\DisplayProof}.
\end{eqnarray*}
Composing these shows $\alpha_{A, \multrue, B} \circ (\id_A \muland \lambda_B)$ is
$$\scpf\frames{
\AxiomC{}\rulename{Id}\UnaryInfC{$A \infers A$}
\AxiomC{}\rulename{Id}\UnaryInfC{$\mt \infers \mt$}
\AxiomC{}\rulename{Id}\UnaryInfC{$B \infers B$}
\rulename{\RO}\BinaryInfC{$\mt,B \infers \mt \muland B$}
\rulename{\RO}\BinaryInfC{$A, \mt, B \infers A \muland (\mt \muland B)$}
\rulename{\LO \pfline\label{pf:1}}\UnaryInfC{$A \muland \mt, B \infers A \muland (\mt \muland B)$}
\rulename{\LO \pfline\label{pf:2}}\UnaryInfC{$(A \muland \mt) \muland B \infers A \muland (\mt \muland B)$}
\AxiomC{}\rulename{Id}\UnaryInfC{$A \infers A$}
\AxiomC{}\rulename{Id}\UnaryInfC{$B \infers B$}
\rulename{\LI}\UnaryInfC{$\mt, B \infers B$}
\rulename{\LO}\UnaryInfC{$\mt \ma B \infers B$}
\rulename{\RO}\BinaryInfC{$A, \mt \ma B \infers A \ma B$}
\rulename{\LO}\UnaryInfC{$A \ma (\mt \ma B) \infers A \ma B$}
\rulename{Cut \pfline\label{pf:3}}\BinaryInfC{$(A \ma \mt) \ma B \infers A \ma B$}\DisplayProof}.$$
Moving the cut~\pfref{pf:3} past rules~\pfref{pf:1} and~\pfref{pf:2} is straightforward using proof transformation (\ref{equation:system-T:l-and-cut1}), which reduces our expression to
$$\frames{%
\AxiomC{}\rulename{Id}\UnaryInfC{$A \infers A$}
\AxiomC{}\rulename{Id}\UnaryInfC{$\mt \infers \mt$}\AxiomC{}
\rulename{Id}\UnaryInfC{$B \infers B$}
\rulename{\RO}\BinaryInfC{$\mt,B \infers \mt \muland B$}
\rulename{\RO \pfline\label{pf:5}}\BinaryInfC{$A, \mt, B \infers A \muland (\mt \muland B)$}
\AxiomC{}\rulename{Id}\UnaryInfC{$A \infers A$}
\AxiomC{}\rulename{Id}\UnaryInfC{$B \infers B$}
\rulename{\LI}\UnaryInfC{$\mt, B \infers B$}
\rulename{\LO}\UnaryInfC{$\mt \ma B \infers B$}
\rulename{\RO}\BinaryInfC{$A, \mt \ma B \infers A \ma B$}
\rulename{\LO \pfline \label{pf:6}}\UnaryInfC{$A \ma (\mt \ma B) \infers A \ma B$}
\rulename{Cut \pfline\label{pf:4}}\BinaryInfC{$A, \mt, B \infers A \ma B$}
\rulename{\LO}\UnaryInfC{$A \muland \mt, B \infers A \ma B$}
\rulename{\LO}\UnaryInfC{$(A \muland \mt) \muland B \infers A \ma B$}\DisplayProof}.$$
Now the cut \pfref{pf:4} with rules \pfref{pf:5} and \pfref{pf:6} can be reduced by the cut-decomposition transformation (\ref{equation:system-T:decomposing-cut}). This reduces the prooftree to
$$\frames{%
\AxiomC{}\rulename{Id \pfline\label{pf:7}}\UnaryInfC{$A \infers A$}%
\AxiomC{}\rulename{Id}\UnaryInfC{$\mt \infers \mt$}
\AxiomC{}\rulename{Id}\UnaryInfC{$B \infers B$}%
\rulename{\RO}\BinaryInfC{$\mt,B \infers \mt \muland B$}%
\AxiomC{}\rulename{Id}\UnaryInfC{$A \infers A$}
\AxiomC{}\rulename{Id}\UnaryInfC{$B \infers B$}
\rulename{\LI}\UnaryInfC{$\mt, B \infers B$}
\rulename{\LO \pfline\label{pf:10}}\UnaryInfC{$\mt \ma B \infers B$}
\rulename{\RO \pfline\label{pf:11}}\BinaryInfC{$A, \mt \ma B \infers A \ma B$}
\rulename{Cut}\BinaryInfC{$A, \mt, B \infers A \ma B$}
\rulename{Cut \pfline\label{pf:8}}\BinaryInfC{$A, \mt, B \infers A \ma B$}
\rulename{\LO}\UnaryInfC{$A \muland \mt, B \infers A \ma B$}
\rulename{\LO}\UnaryInfC{$(A \muland \mt) \muland B \infers A \ma B$}\DisplayProof}.$$
Now we eliminate the redundant identity \pfref{pf:7} and cut rule \pfref{pf:8} with proof transformation (\ref{equation:system-T:left-identity-transformation}). We also commute rules \pfref{pf:10} and \pfref{pf:11} using transformation (\ref{equation:system-T:left-right-and-transformation}). This reduces the form for $\alpha_{A, \multrue, B} \circ (\id_A \muland \lambda_B)$ to
$$\frames{%
\AxiomC{}\rulename{Id}\UnaryInfC{$\mt \infers \mt$}
\AxiomC{}\rulename{Id}\UnaryInfC{$B \infers B$}%
\rulename{\RO \pfline\label{pf:13}}\BinaryInfC{$\mt,B \infers \mt \muland B$}%
\AxiomC{}\rulename{Id}\UnaryInfC{$A \infers A$}
\AxiomC{}\rulename{Id}\UnaryInfC{$B \infers B$}
\rulename{\LI}\UnaryInfC{$\mt, B \infers B$}
\rulename{\RO}\BinaryInfC{$A, \mt, B \infers A \ma B$}
\rulename{\LO \pfline\label{pf:14}}\UnaryInfC{$A, \mt \ma B \infers B$}
\rulename{Cut \pfline\label{pf:12}}\BinaryInfC{$A, \mt, B \infers A \ma B$}
\rulename{\LO}\UnaryInfC{$A \muland \mt, B \infers A \ma B$}
\rulename{\LO}\UnaryInfC{$(A \muland \mt) \muland B \infers A \ma B$}\DisplayProof}.$$
Again we use cut decomposition (\ref{equation:system-T:decomposing-cut}), now on cut \pfref{pf:12} with \pfref{pf:13} and \pfref{pf:14}. This produces
$$\frames{%
\AxiomC{}\rulename{Id \pfline\label{pf:15}}\UnaryInfC{$\mt \infers \mt$}
\AxiomC{}\rulename{Id \pfline\label{pf:16}}\UnaryInfC{$B \infers B$}
\AxiomC{}\rulename{Id}\UnaryInfC{$A \infers A$}
\AxiomC{}\rulename{Id}\UnaryInfC{$B \infers B$}
\rulename{\LI}\UnaryInfC{$\mt, B \infers B$}
\rulename{\RO}\BinaryInfC{$A, \mt, B \infers A \ma B$}
\rulename{Cut \pfline\label{pf:17}}\BinaryInfC{$A, \mt \ma B \infers B$}
\rulename{Cut \pfline\label{pf:18}}\BinaryInfC{$A, \mt, B \infers A \ma B$}
\rulename{\LO}\UnaryInfC{$A \muland \mt, B \infers A \ma B$}
\rulename{\LO}\UnaryInfC{$(A \muland \mt) \muland B \infers A \ma B$}\DisplayProof}.$$
Again we may eliminate the redundant identitities \pfref{pf:15}, \pfref{pf:16} and cut rules \pfref{pf:17}, \pfref{pf:18} to give
$$\frames{%
\AxiomC{}\rulename{Id}\UnaryInfC{$A \infers A$}
\AxiomC{}\rulename{Id}\UnaryInfC{$B \infers B$}
\rulename{\LI \pfline\label{pf:19}}\UnaryInfC{$\mt, B \infers B$}
\rulename{\RO \pfline\label{pf:20}}\BinaryInfC{$A, \mt, B \infers A \ma B$}
\rulename{\LO}\UnaryInfC{$A \muland \mt, B \infers A \ma B$}
\rulename{\LO}\UnaryInfC{$(A \muland \mt) \muland B \infers A \ma B$}\DisplayProof}.$$
Finally we commute rules \pfref{pf:19} and \pfref{pf:20} using (\ref{equation:system-T:left-1-right-and-transformation}) to obtain
$$\frames{%
\AxiomC{}\rulename{Id}\UnaryInfC{$A \infers A$}
\AxiomC{}\rulename{Id}\UnaryInfC{$B \infers B$}
\rulename{\RO}\BinaryInfC{$A, B \infers A \ma B$}
\rulename{\LI}\UnaryInfC{$A, \mt, B \infers A \ma B$}
\rulename{\LO}\UnaryInfC{$A \muland \mt, B \infers A \ma B$}
\rulename{\LO}\UnaryInfC{$(A \muland \mt) \muland B \infers A \ma B$}\DisplayProof}.$$

Yet by definition
$$\rho_A \muland \id_B = \frames{%
\AxiomC{}\rulename{Id}\UnaryInfC{$A \infers A$}
\rulename{\LI \pfline\label{pf:23}}\UnaryInfC{$A, \mt \infers A$}
\rulename{\LO \pfline\label{pf:22}}\UnaryInfC{$A \ma \mt \infers A$}
\AxiomC{}\rulename{Id}\UnaryInfC{$B \infers B$}
\rulename{\RO \pfline\label{pf:21}}\BinaryInfC{$A \ma \mt, B \infers A \ma B$}
\rulename{\LO}\UnaryInfC{$(A \ma \mt) \ma B \infers A \ma B$}\DisplayProof}.$$
Using transformations (\ref{equation:system-T:left-right-and-transformation}) and (\ref{equation:system-T:left-1-right-and-transformation}) to move rule \pfref{pf:21} past \pfref{pf:22} and then \pfref{pf:23} produces exactly the same expression as the reduced form for $\alpha_{A, \multrue, B} \circ (\id_A \muland \lambda_B)$. Therefore $\alpha_{A, \multrue, B} \circ (\id_A \muland \lambda_B) = \rho_A \muland \id_B$ as required.
\end{proof}

\begin{lemma}[Pentagon Rule]
Let logic system $\sysT'$ give a category $\catC$ as discussed above  and suppose that proofs related by transformation rules (\ref{equation:system-T:decomposing-cut}, \ref{equation:system-T:l-and-cut1}, \ref{equation:system-T:left-identity-transformation}) or their inverses are equivalent.. Then, the pentagon coherence relation is satisfied:
\begin{equation*}
        \begin{tikzpicture}[baseline=(current bounding box.center)]
            \draw (1.5,3) node (n0) {$(A \monoidal (B \monoidal C)) \monoidal D$};
            \draw (0,1.5) node (n1) {$((A \monoidal B) \monoidal C) \monoidal D$};
            \draw (4,0) node (n2) {$(A \monoidal B) \monoidal (C \monoidal D)$};
            \draw (6.5,3) node (n3) {$A \monoidal ((B \monoidal C) \monoidal D)$};
            \draw (8,1.5) node (n4) {$A \monoidal (B \monoidal (C \monoidal D))$};
            \draw (4,3.2) node (a0) {\scriptsize $\alpha_{A,B\monoidal C,D}$};
            \draw (-0.3,2.35) node (a1) {\scriptsize $\alpha_{A,B,C}\monoidal \id_D$};
            \draw (8.3,2.35) node (a2) {\scriptsize $\id_A\monoidal\alpha_{B,C,D}$};
            \draw (1,0.6) node (a3) {\scriptsize $\alpha_{A\monoidal B,C,D}$};
            \draw (7,0.6) node (a4) {\scriptsize $\alpha_{A,B,C\monoidal D}$};
            \draw[->] (n0) -- (n1);
            \draw[->] (n1) -- (n2);
            \draw[->] (n0) -- (n3);
            \draw[->] (n3) -- (n4);
            \draw[->] (n4) -- (n2);
        \end{tikzpicture}
    \end{equation*}
\end{lemma}

\scpf
\begin{proof}
Let $A, B, C, D$ be objects; we need to show that:
\begin{equation}\label{eq:pentagon}
	 (\alpha_{A,B,C} \monoidal \id_D) \circ \alpha_{A, B \monoidal C, D} \circ (\id_A \monoidal \alpha_{B,C,D}) = \alpha_{A \monoidal B, C, D} \circ \alpha_{A, B, C \monoidal D}.
\end{equation}
Clearly the proofs that map to the morphisms above involve multiple cuts and it will be shown that both sides of the equation reduce to the same equivalent form after eliminating all the cuts in the proofs. Starting with the morphisms on the right side of this equation, the proof of interest is:

\footnotesize
\begin{prooftree}
\ruleId{$A$}
\ruleId{$B$}
\rulename{\RO}\BinaryInfC{$A, B \infers A \ma B$}
\ruleId{$C$}
\ruleId{$D$}
\rulename{\RO}\BinaryInfC{$C, D \infers C \ma D$}
\rulename{\RO\pfline\label{pf:4a}}\BinaryInfC{$A, B, C, D \infers (A \ma B) \ma (C \ma D)$}
\rulename{\LO $\times 3$ \pfline\label{pf:1a}}\UnaryInfC{$((A \ma B) \ma C) \ma D \infers (A \ma B) \ma (C \ma D)$}%
\ruleId{$A$}
\ruleId{$B$}
\ruleId{$C$}
\ruleId{$D$}
\rulename{\RO}\BinaryInfC{$C, D \infers C \ma D$}
\rulename{\RO}\BinaryInfC{$B, C, D \infers B \ma (C \ma D)$}
\rulename{\RO}\BinaryInfC{$A, B, C, D \infers A \ma (B \ma (C \ma D))$}
\rulename{\LO $\times 2$}\UnaryInfC{$(A \ma B), (C \ma D) \infers A \ma (B \ma (C \ma D))$}
\rulename{\LO \pfline\label{pf:3a}}\UnaryInfC{$(A \ma B) \ma (C \ma D) \infers A \ma (B \ma (C \ma D))$}%
\rulename{Cut \pfline\label{pf:2a}}\BinaryInfC{$((A \ma B) \ma C) \ma D \infers A \ma (B \ma (C \ma D))$.}
\end{prooftree}

\normalsize
Moving the cut \pfref{pf:2a} past the three left-$\ma$ introductions \pfref{pf:1a} is straightforward using (\ref{equation:system-T:l-and-cut1}). Applying cut decomposition (\ref{equation:system-T:decomposing-cut}) to \pfref{pf:4a} and \pfref{pf:3a} transforms the resulting proof to

\footnotesize\begin{prooftree}
\ruleId{$A$}
\ruleId{$B$}
\rulename{\RO \pfline\label{pf:9a}}\BinaryInfC{$A, B \infers A \ma B$}
\ruleId{$C$}
\ruleId{$D$}
\rulename{\RO \pfline\label{pf:8a}}\BinaryInfC{$C, D \infers C \ma D$}
\ruleId{$A$}
\ruleId{$B$}
\ruleId{$C$}
\ruleId{$D$}
\rulename{\RO}\BinaryInfC{$C, D \infers C \ma D$}
\rulename{\RO}\BinaryInfC{$B, C, D \infers B \ma (C \ma D)$}
\rulename{\RO}\BinaryInfC{$A, B, C, D \infers A \ma (B \ma (C \ma D))$}
\rulename{\LO $\times 2$ \pfline\label{pf:7a}}\UnaryInfC{$A \ma B, C \ma D \infers A \ma (B \ma (C \ma D))$}
\rulename{Cut \pfline\label{pf:6a}}\BinaryInfC{$A \ma B, C, D \infers A \ma (B \ma (C \ma D))$}
\rulename{Cut \pfline\label{pf:5a}}\BinaryInfC{$A, B, C, D \infers A \ma (B \ma (C \ma D))$}
\rulename{\LO $\times 3$}\UnaryInfC{$((A \ma B) \ma C) \ma D \infers A \ma (B \ma (C \ma D))$.}%
\end{prooftree}

\normalsize
Next one applies cut decomposition to \pfref{pf:6a}, \pfref{pf:8a} and \pfref{pf:7a}, followed by moving cut \pfref{pf:5a} across \pfref{pf:9a} and the rest of the proof. Finally four applications of identity transformation (\ref{equation:system-T:left-identity-transformation}) reduces this proof to a normal form

\begin{equation}\label{pft:nf}
\footnotesize
\ruleId{$A$}
\ruleId{$B$}
\ruleId{$C$}
\ruleId{$D$}
\rulename{\RO}\BinaryInfC{$C, D \infers C \ma D$}
\rulename{\RO}\BinaryInfC{$B, C, D \infers (B \ma (C \ma D))$}
\rulename{\RO}\BinaryInfC{$A, B, C, D \infers A \ma (B \ma (C \ma D))$}
\rulename{\LO}\UnaryInfC{$A \ma B, C, D \infers A \ma (B \ma (C \ma D))$}
\rulename{\LO}\UnaryInfC{$(A \ma B) \ma C, D \infers A \ma (B \ma (C \ma D))$}
\rulename{\LO}\UnaryInfC{$((A \ma B) \ma C) \ma D \infers A \ma (B \ma (C \ma D))$.}\DisplayProof
\end{equation}

Now we must reduce the left side of (\ref{eq:pentagon}) to this normal form We start by simplifying the first composition $\alpha_{A,B,C} \ma \id_D \circ \alpha_{A, B \ma C, D}$, which is defined by the proof
\scpf
$$\footnotesize
\AxiomC{}\rulename{Id}\UnaryInfC{$A \infers A$}
\AxiomC{}\rulename{Id}\UnaryInfC{$B \infers B$}
\AxiomC{}\rulename{Id}\UnaryInfC{$C \infers C$}
\rulename{\RO}\BinaryInfC{$B, C \infers B \ma C$}
\rulename{\RO}\BinaryInfC{$A, B, C \infers A \ma (B \ma C)$}
\rulename{\LO}\UnaryInfC{$A \ma B, C \infers A \ma (B \ma C)$}
\rulename{\LO}\UnaryInfC{$(A \ma B) \ma C \infers A \ma (B \ma C)$}
\AxiomC{}\rulename{Id}\UnaryInfC{$D \infers D$}
\rulename{\RO \pfline\label{pf:4b}}\BinaryInfC{$(A \ma B) \ma C, D \infers (A \ma (B \ma C)) \ma D$}
\rulename{\LO \pfline\label{pf:3b}}\UnaryInfC{$((A \ma B) \ma C) \ma D \infers (A \ma (B \ma C)) \ma D$}%
\ruleId{$A$}
\ruleId{$B$}
\ruleId{$C$}
\rulename{\RO}\BinaryInfC{$B, C \infers B \ma C$}
\ruleId{$D$}
\rulename{\RO}\BinaryInfC{$B, C, D \infers (B \ma C) \ma D$}
\rulename{\RO}\BinaryInfC{$A, B, C, D \infers A \ma ((B \ma C) \ma D)$}
\rulename{\LO}\UnaryInfC{$A, B \ma C, D \infers A \ma ((B \ma C) \ma D)$}
\rulename{\LO}\UnaryInfC{$A \ma (B \ma C), D \infers A \ma ((B \ma C) \ma D)$}
\rulename{\LO \pfline\label{pf:2b}}\UnaryInfC{$(A \ma (B \ma C)) \ma D \infers A \ma ((B \ma C) \ma D)$}
\rulename{Cut \pfline\label{pf:1b}}\BinaryInfC{$((A \ma B) \ma C) \ma D \infers A \ma ((B \ma C) \ma D)$.}\DisplayProof
$$
\normalsize
To simplify this, we first commute the cut \pfref{pf:1b} past \pfref{pf:3b} with (\ref{equation:system-T:l-and-cut1}) and then use cut decomposition (\ref{equation:system-T:decomposing-cut}) on \pfref{pf:4b} and \pfref{pf:2b}. This transforms this proof into
$$\footnotesize
\AxiomC{}\rulename{Id}\UnaryInfC{$A \infers A$}
\AxiomC{}\rulename{Id}\UnaryInfC{$B \infers B$}
\AxiomC{}\rulename{Id}\UnaryInfC{$C \infers C$}
\rulename{\RO}\BinaryInfC{$B, C \infers B \ma C$}
\rulename{\RO \pfline\label{pf:11b}}\BinaryInfC{$A, B, C \infers A \ma (B \ma C)$}
\rulename{\LO \pfline\label{pf:10b}}\UnaryInfC{$A \ma B, C \infers A \ma (B \ma C)$}
\rulename{\LO \pfline\label{pf:9b}}\UnaryInfC{$(A \ma B) \ma C \infers A \ma (B \ma C)$}
\rulename{Id \pfline\label{pf:8b}}\AxiomC{}\UnaryInfC{$D \infers D$}
\ruleId{$A$}
\ruleId{$B$}
\ruleId{$C$}
\rulename{\RO}\BinaryInfC{$B, C \infers B \ma C$}
\ruleId{$D$}
\rulename{\RO}\BinaryInfC{$B, C, D \infers (B \ma C) \ma D$}
\rulename{\RO}\BinaryInfC{$A, B, C, D \infers A \ma ((B \ma C) \ma D)$}
\rulename{\LO}\UnaryInfC{$A, B \ma C, D \infers A \ma ((B \ma C) \ma D)$}
\rulename{\LO  \pfline\label{pf:7b}}\UnaryInfC{$A \ma (B \ma C), D \infers A \ma ((B \ma C) \ma D)$}
\rulename{Cut \pfline\label{pf:6b}}\BinaryInfC{$(A \ma (B \ma C)), D \infers A \ma ((B \ma C) \ma D)$}
\rulename{Cut \pfline\label{pf:5b}}\BinaryInfC{$((A \ma B) \ma C), D \infers A \ma ((B \ma C) \ma D)$}
\rulename{\LO}\UnaryInfC{$((A \ma B) \ma C) \ma D \infers (A \ma (B \ma C)) \ma D$.}%
\DisplayProof
$$\normalsize
Now we eliminate the identity-cut pair \pfref{pf:8b} and \pfref{pf:6b} using (\ref{equation:system-T:left-identity-transformation}) and commute the cut \pfref{pf:5b} past \pfref{pf:9b} and \pfref{pf:10b} using (\ref{equation:system-T:l-and-cut1}). This allows us to apply cut decomposition (\ref{equation:system-T:decomposing-cut}) on \pfref{pf:11b} and \pfref{pf:7b}. This leads to the proof
$$\footnotesize
\AxiomC{}\rulename{Id \pfline\label{pf:17b}}\UnaryInfC{$A \infers A$}
\AxiomC{}\rulename{Id \pfline\label{pf:16b}}\UnaryInfC{$B \infers B$}
\AxiomC{}\rulename{Id \pfline\label{pf:15b}}\UnaryInfC{$C \infers C$}
\rulename{\RO \pfline\label{pf:14b}}\BinaryInfC{$B, C \infers B \ma C$}
\ruleId{$A$}
\ruleId{$B$}
\ruleId{$C$}
\rulename{\RO}\BinaryInfC{$B, C \infers B \ma C$}
\ruleId{$D$}
\rulename{\RO}\BinaryInfC{$B, C, D \infers (B \ma C) \ma D$}
\rulename{\RO}\BinaryInfC{$A, B, C, D \infers A \ma ((B \ma C) \ma D)$}
\rulename{\LO \pfline\label{pf:13b}}\UnaryInfC{$A, B \ma C, D \infers A \ma ((B \ma C) \ma D)$}
\rulename{Cut \pfline\label{pf:12b}}\BinaryInfC{$A, B, C, D \infers A \ma ((B \ma C) \ma D)$}
\rulename{Cut}\BinaryInfC{$A, B, C, D \infers A \ma ((B \ma C) \ma D)$}
\rulename{\LO}\UnaryInfC{$A \ma B, C, D \infers A \ma (B \ma C)$}
\rulename{\LO}\UnaryInfC{$(A \ma B) \ma C, D \infers A \ma (B \ma C)$}
\rulename{\LO}\UnaryInfC{$((A \ma B) \ma C) \ma D \infers (A \ma (B \ma C)) \ma D$.}%
\DisplayProof
$$\normalsize
Performing cut decomposition again at \pfref{pf:14b}, \pfref{pf:13b}, and \pfref{pf:12b}, and then eliminating the cuts with the identities \pfref{pf:17b}, \pfref{pf:16b}, and \pfref{pf:15b} transforms $\frames{\alpha_{A,B,C} \monoidal \id_D \circ \alpha_{A, B \monoidal C, D}}$ into the cut-free proof
\begin{equation}\label{equation:pentagon:normal-form}
\ruleId{$A$}
\ruleId{$B$}
\ruleId{$C$}
\rulename{\RO}\BinaryInfC{$B, C \infers B \ma C$}
\ruleId{$D$}
\rulename{\RO}\BinaryInfC{$B, C, D \infers (B \ma C) \ma D$}
\rulename{\RO}\BinaryInfC{$A, B, C, D \infers A \ma ((B \ma C) \ma D)$}
\rulename{\LO}\UnaryInfC{$(A \ma B), C, D \infers A \ma ((B \ma C) \ma D)$}
\rulename{\LO}\UnaryInfC{$(A \ma B) \ma C, D \infers A \ma ((B \ma C) \ma D)$}
\rulename{\LO}\UnaryInfC{$((A \ma B) \ma C) \ma D \infers A \ma ((B \ma C) \ma D)$.}\DisplayProof
\end{equation}

Finally, the left side of (\ref{eq:pentagon}) has the simplified proof
$$\scpf\footnotesize
\ruleId{$A$}
\ruleId{$B$}
\ruleId{$C$}
\rulename{\RO}\BinaryInfC{$B, C \infers B \ma C$}
\ruleId{$D$}
\rulename{\RO}\BinaryInfC{$B, C, D \infers (B \ma C) \ma D$}
\rulename{\RO \pfline\label{pf:1c}}\BinaryInfC{$A, B, C, D \infers A \ma ((B \ma C) \ma D)$}
\rulename{\LO \pfline\label{pf:2c}}\UnaryInfC{$(A \ma B), C, D \infers A \ma ((B \ma C) \ma D)$}
\rulename{\LO \pfline\label{pf:3c}}\UnaryInfC{$(A \ma B) \ma C, D \infers A \ma ((B \ma C) \ma D)$}
\rulename{\LO \pfline\label{pf:4c}}\UnaryInfC{$((A \ma B) \ma C) \ma D \infers A \ma ((B \ma C) \ma D)$}%
\ruleId{$A$}
\ruleId{$B$}
\ruleId{$C$}
\ruleId{$D$}
\rulename{\RO}\BinaryInfC{$C, D \infers C \ma D$}
\rulename{\RO}\BinaryInfC{$B, C, D \infers B \ma (C \ma D)$}
\rulename{\LO}\UnaryInfC{$B \ma C, D \infers B \ma (C \ma D)$}
\rulename{\LO}\UnaryInfC{$(B \ma C) \ma D \infers B \ma (C \ma D)$}
\rulename{\RO  \pfline\label{pf:7c}}\BinaryInfC{$A, (B \ma C) \ma D \infers A \ma (B \ma (C \ma D))$}
\rulename{\LO \pfline\label{pf:5c}}\UnaryInfC{$A \ma ((B \ma C) \ma D) \infers A \ma (B \ma (C \ma D))$}%
\rulename{Cut \pfline\label{pf:6c}}\BinaryInfC{$((A \ma B) \ma C) \ma D \infers A \ma (B \ma (C \ma D))$}
\DisplayProof
$$\normalsize
First moving the cut \pfref{pf:6c} past \pfref{pf:4c}, \pfref{pf:2c} and \pfref{pf:3c} allows us to use cut decomposition on \pfref{pf:1c} and \pfref{pf:5c}. Eliminating the redundant identity and cut, and commuting the remaining cut through a right-$\ma$ introduction \pfref{pf:7c} reduces the proof to
$$\footnotesize
\AxiomC{}\rulename{Id \pfline\label{pf:12c}}\UnaryInfC{$A \infers A$}
\AxiomC{}\rulename{Id \pfline\label{pf:13c}}\UnaryInfC{$B \infers B$}
\AxiomC{}\rulename{Id \pfline\label{pf:14c}}\UnaryInfC{$C \infers C$}
\rulename{\RO \pfline\label{pf:15c}}\BinaryInfC{$B, C \infers B \ma C$}
\ruleId{$D$}
\rulename{\RO \pfline\label{pf:8c}}\BinaryInfC{$B, C, D \infers (B \ma C) \ma D$}
\ruleId{$B$}
\ruleId{$C$}
\ruleId{$D$}
\rulename{\RO}\BinaryInfC{$C, D \infers C \ma D$}
\rulename{\RO}\BinaryInfC{$B, C, D \infers B \ma (C \ma D)$}
\rulename{\LO \pfline\label{pf:9c}}\UnaryInfC{$B \ma C, D \infers B \ma (C \ma D)$}
\rulename{\LO \pfline\label{pf:10c}}\UnaryInfC{$(B \ma C) \ma D \infers B \ma (C \ma D)$}%
\rulename{Cut \pfline\label{pf:11c}}\BinaryInfC{$B, C, D \infers (B \ma (C \ma D))$}
\rulename{\RO}\BinaryInfC{$A, B, C, D \infers A \ma (B \ma (C \ma D))$}
\rulename{\LO $\times 3$}\UnaryInfC{$((A \ma B) \ma C) \ma D \infers A \ma (B \ma (C \ma D)).$}\DisplayProof$$
Finally, we apply two cut decompositions. The first cut \pfref{pf:11c} decomposes against \pfref{pf:8c} and \pfref{pf:10c}. The higher of the resulting cuts decomposes with \pfref{pf:15c} and \pfref{pf:9c}. This produces three cuts which are then eliminated with the identities \pfref{pf:12c}, \pfref{pf:13c}, and \pfref{pf:14c}. The resulting cut-free proof is
$$\ruleId{$A$}
\ruleId{$B$}
\ruleId{$C$}
\ruleId{$D$}
\rulename{\RO}\BinaryInfC{$C, D \infers C \ma D$}
\rulename{\RO}\BinaryInfC{$B, C, D \infers B \ma (C \ma D)$}
\rulename{\RO}\BinaryInfC{$A, B, C, D \infers A \ma (B \ma (C \ma D))$}
\rulename{\LO $\times 3$ }\UnaryInfC{$((A \ma B) \ma C) \ma D \infers A \ma (B \ma (C \ma D)),$}\DisplayProof$$
which coincides precisely with the proof (\ref{equation:pentagon:normal-form}).
\end{proof}

This chain of lemmas proves the following result.

\begin{theorem}
    Let $\Phi$ be a countable set of atomic propositions and $\system{T}'$ be the associated logic. Define the category $\category{C}$ as above and assume that proofs related by transformations 
    (\ref{equation:system-T:vertical-cut-commutivity}, \ref{equation:system-T:decomposing-cut},
    \ref{equation:system-T:removing-1-cut},
    \ref{equation:system-T:l-and-cut1}, 
    \ref{equation:system-T:r-and-cut}, \ref{equation:system-T:right-identity-transformation}, \ref{equation:system-T:left-identity-transformation}, \ref{equation:system-T:left-right-and-transformation}, \ref{equation:system-T:left-1-right-and-transformation}) or their inverses are equivalent. Then $\category{C}$ is a monoidal category.
\end{theorem}


For $\sysT$ we additionally have the Exchange rule, which allows us to reorder the terms in the antecedent of any inference. We consider use of the exchange rule as formalism, and hence to not introduce a set of proof transformation associated to using it, in our following proof we will be explicit about its use to better illustrate the role of exchange. Using this symmetry we define another natural isomorphism $\sigma$ just as in the proof of Proposition \ref{proposition:system-T:algebra}
$$\sigma_{A,B} = \frames{%
\rulename{Id}\AxiomC{}\UnaryInfC{$B \infers B$}
\rulename{Id}\AxiomC{}\UnaryInfC{$A \infers A$}
\rulename{R-$\muland$}\BinaryInfC{$B, A \infers B\ma A$}
\rulename{Exchange}\UnaryInfC{$A, B \infers B \ma A$}
\rulename{L-$\muland$}\UnaryInfC{$A\ma B \infers B \ma A$.}
\DisplayProof}
\in \hom(A \monoidal B, B \monoidal A).$$
This must satisfy the functorial property
\begin{equation*}
\begin{tikzcd}
 A \monoidal B \arrow[rr, "\sigma_{A,B}"] \arrow[dd, "\frames{\Pi_1} \monoidal \frames{\Pi_2}"'] & & B \monoidal A \arrow[dd, "\frames{\Pi_2} \monoidal \frames{\Pi_1}"] \\ 
 & & \\
 C \monoidal D \arrow[rr, "\sigma_{C,D}"'] & & D \monoidal C
\end{tikzcd}
\end{equation*}  
which is comparatively straightforward and so is left to the reader.

\begin{lemma}
    For $\sysT$, the following commutative diagram holds:
    $$\begin{tikzcd}
            A\monoidal B \arrow[rr, "\id_{A\monoidal B}"'] \arrow[rd, "\sigma_{A,B}"'] & & A \monoidal B. \\
            & B\monoidal A \arrow[ru, "\sigma_{B,A}"']
    \end{tikzcd}$$
    In particular, $\sigma$ is a natural isomorphism.
\end{lemma}
\begin{proof}
We need to show that
\begin{equation}
\label{eq:inverse}
\sigma_{A,B} \circ \sigma_{B,A} = \id_{A \monoidal B}.
\end{equation}
Note that as the corresponding term for $A \monoidal B$ is $A \ma B$ which is not a single atomic proposition, the identity rule for $A \monoidal B$ is equivalent to composing the tensor of the identity rules for the individual propositions making up the term i.e. for $A$ and $B$ in this case as shown below.
\begin{equation*}
\id_{A \monoidal B} = \id_{A} \monoidal \id_B = \frames{%
\ruleid{$A$}
\ruleid{$B$}
\rulename{\RO}\BinaryInfC{$A, B \infers A \ma B$}
\rulename{\LO}\UnaryInfC{$A \ma B \infers A \ma B$}\DisplayProof.}
\end{equation*}

Now consider $\sigma_{A,B} \circ \sigma_{B,A}$. By construction the associated proof is
\begin{prooftree}
\scpf
\ruleid{$B$}
\ruleid{$A$}
\rulename{\RO}
\BinaryInfC{$B, A \infers B \ma A$}
\rulename{Ex \pfline\label{pf:1e}}\UnaryInfC{$A, B \infers B \ma A$}
\rulename{\LO \pfline\label{pf:2e}}\UnaryInfC{$A \ma B \infers B \ma A$}
\ruleid{$A$}
\ruleid{$B$}
\rulename{\RO}\BinaryInfC{$A, B \infers A \ma B$}
\rulename{Ex}\UnaryInfC{$B, A \infers A \ma B$}\rulename{\LO \pfline\label{pf:3e}}\UnaryInfC{$B \ma A \infers A \ma B$}
\rulename{Cut \pfline\label{pf:4e}}\BinaryInfC{$A \ma B \infers A \ma B$}
\end{prooftree}
Moving the cut~\pfref{pf:4e} past~\pfref{pf:2e} and~\pfref{pf:1e} directly gives 
\begin{prooftree}
\ruleid{$B$}
\ruleid{$A$}
\rulename{\RO \pfline\label{pf:7e}}\BinaryInfC{$B, A \infers B \ma A$}
\ruleid{$A$}
\ruleid{$B$}
\rulename{\RO}\BinaryInfC{$A, B \infers A \ma B$}
\rulename{Ex}\UnaryInfC{$B, A \infers A \ma B$}
\rulename{\LO \pfline\label{pf:6e}}\UnaryInfC{$B \ma A \infers A \ma B$}
\rulename{Cut \pfline\label{pf:5e}}\BinaryInfC{$B, A \infers A \ma B$}
\rulename{Ex}\UnaryInfC{$A, B \infers A \ma B$}
\rulename{\LO}\UnaryInfC{$A \ma B \infers B \ma A$}
\end{prooftree}
Now we apply cut decomposition to \pfref{pf:5e} using \pfref{pf:7e} and \pfref{pf:6e}, and remove the redundant identities and cuts. This gives the proof
\begin{prooftree}
\ruleid{$A$}\ruleid{$B$}\rulename{\RO}\BinaryInfC{$A, B \infers A \ma B$}\rulename{Ex \pfline\label{pf:8e}}\UnaryInfC{$B, A \infers A \ma B$}\rulename{Ex \pfline\label{pf:9e}}\UnaryInfC{$A, B \infers A \ma B$}\rulename{\LO}\UnaryInfC{$A \ma B \infers B \ma A$}
\end{prooftree}
Notice that both uses of the exchange rules \pfref{pf:8e} and \pfref{pf:9e} are purely formalism and cancel producing the same proof as $\id_{A \monoidal B}$ above.
\end{proof}

\begin{lemma}
    For $\sysT$, the following commutative diagram holds:
    $$\begin{tikzcd}
            A\monoidal I \arrow[rr, "\sigma_{A,I}"'] \arrow[rd, "\rho_A"'] & & I \monoidal A \arrow[ld, "\lambda_A"]\\
            & A
    \end{tikzcd}$$
\end{lemma}
\begin{proof}
We begin with the expression
$$\scpf
\sigma_{A,I} \circ \lambda _A = \frames{
\rulename{Id}\AxiomC{}\UnaryInfC{$\mt \infers \mt$}
\rulename{Id}\AxiomC{}\UnaryInfC{$A \infers A$}
\rulename{\RO  \pfline\label{pf:5d}}\BinaryInfC{$\mt, A \infers \mt \ma A$}
\rulename{Ex \pfline\label{pf:4d}}\UnaryInfC{$A, \mt \infers \mt \ma A$}
\rulename{\LO \pfline\label{pf:1d}}\UnaryInfC{$A \ma \mt \infers \mt \ma A$}
\rulename{Id}\AxiomC{}\UnaryInfC{$A \infers A$}
\rulename{\LI}\UnaryInfC{$\mt, A \infers A$}
\rulename{\LO \pfline\label{pf:2d}}\UnaryInfC{$\mt \ma A \infers A$}
\rulename{Cut \pfline\label{pf:3d}}\BinaryInfC{$A \ma \mt \infers A$.}\DisplayProof}
$$
Commuting the cut \pfref{pf:3d} through the left-$\ma$ introduction \pfref{pf:1d} with (\ref{equation:system-T:l-and-cut1}) and trivially through the exchange \pfref{pf:4d}, we then perform cut decomposition (\ref{equation:system-T:decomposing-cut}) with \pfref{pf:5d} and \pfref{pf:2d}. This reduces our form to
$$\sigma_{A,I} \circ \lambda _A = \frames{
\rulename{Id}\AxiomC{}\UnaryInfC{$\mt \infers \mt$}
\rulename{Id}\AxiomC{}\UnaryInfC{$A \infers A$}
\rulename{Id}\AxiomC{}\UnaryInfC{$A \infers A$}
\rulename{\LI}\UnaryInfC{$\mt, A \infers A$}
\rulename{Cut}\BinaryInfC{$\mt, A \infers A$}
\rulename{Cut}\BinaryInfC{$\mt, A \infers A$}
\rulename{Ex}\UnaryInfC{$A, \mt \infers A$}
\rulename{\LO}\UnaryInfC{$A \ma \mt \infers A$}
\DisplayProof}
$$
Eliminating the identities and cuts using (\ref{equation:system-T:left-identity-transformation}) gives us
$$\sigma_{A,I} \circ \lambda _A = \frames{
\rulename{Id}\AxiomC{}\UnaryInfC{$A \infers A$}
\rulename{\LI}\UnaryInfC{$\mt, A \infers A$}
\rulename{Ex}\UnaryInfC{$A, \mt \infers A$}
\rulename{\LO}\UnaryInfC{$A \ma \mt \infers A$}
\DisplayProof} = \rho_A.
$$
Note that the slight difference in this expression for $\rho_A$ from equation~\eqref{eq:rho-lambda} comes as we use the system $\sysT$'s rules to define it while the latter follows from the rules of system $\sysT'$. 
\end{proof}

\begin{lemma}[Hexagon rule]
For $\sysT$, the following commutative diagram holds:
    $$\begin{tikzcd}
            & (A\monoidal B) \monoidal C \arrow[rr, "\sigma_{A,B} \monoidal \id_C"] \arrow[ld, "\alpha_{A,B,C}"'] & & (B \monoidal A) \monoidal C \arrow[rd, "\alpha_{B,A,C}"] & \\
            A \monoidal (B \monoidal C) \arrow[rd, "\sigma_{A, B \monoidal C}"'] & & & & B \monoidal (A\monoidal C) \arrow[ld, "\id_B \monoidal \sigma_{A,C}"]\\
            & (B \monoidal C) \monoidal A \arrow[rr, "\alpha_{B,C,A}"'] & & B \monoidal (C \monoidal A) &
    \end{tikzcd}$$
\end{lemma}

As the proof of the Hexagon Rule is similar in form to that of the Pentagon Rule, except with appropriate use of the Exchange rule as needed, this proof is left to the reader.

\begin{theorem}[Coherence theorem]
Let $\Phi$ be a countable set of atomic propositions and $\system{T}$ be the associated logic. Then the monoidal category $\category{C}$ constructed for $\system{T}$ above is a symmetric monoidal category. Moreover the unique functor from the free symmetric monoidal category on $\Phi$ is fully faithful.
\end{theorem}
\begin{proof}
    Recall that the free symmetric monoidal category on $\Phi$, which we denote as $\category{F}(\Phi)$, has as its objects multisets of elements of $\Phi$ with monoidal product given by the union and identity by the empty multiset. Its morphisms are those induced by taking the multiset of identity morphisms on the constituent elements. This is somewhat ill defined when it comes to the empty multiset $\emptyset$ and so we pose a unique ``identity'' morphism $\emptyset^{\monoidal a} \to \emptyset^{\monoidal b}$ for any $a,b \geq 0$. Modulo this, a morphism only exists between objects whose multisets consist of the same atomic propositions.
    
    Clearly the unique functor $\category{F}(\Phi) \to \category{C}$ is just the identity, which is faithful. To show full faithfulness, we need to show that each proof $\Gamma \infers A$ in $\system{T}$ produces a morphism $\frames{\Gamma \infers A}$ that is in the image of this functor. That is, $\frames{\Gamma \infers A} = \id_A$.
    
    This follows from Corollary \ref{corollary:normal-form}. By that result, the inference $\Gamma \infers A$ is just a repackaging of $P_1 \ma \cdots P_m \ma \mt^{\ma a} \infers P_1 \ma \cdots P_m \ma \mt^{\ma b}$ for $P_j \in \Phi$. Any proof of this is unique, after passing to a cut-free proof and then commuting left-$\ma$ and right-$\ma$ and left-$\mt$ and right-$\mt$ into a normal form. Therefore this proof must be the image of the unique morphism in the free category $\category{F}(\Phi)$.
\end{proof}

\begin{example}\label{example:coherence-category}
Continuing examples \ref{example:coherence} and \ref{example:coherence-algebra}, we again defer a full treatment of categorical semantics of polynomial systems over $\sysT$ to later work. Yet, at this point, we note that every atomic proposition of the system is disposable: $Q(\lambda) \infers \mt$ for all $\lambda \in [0,1]$. We, therefore, require that in any categorical model of this system $I = \frames{\mt}$  be terminal i.e., for all $\lambda \in [0,1]$, there exists a single morphism in $\hom( Q(\lambda), I)$.

Similarly, $Q(1) \infers Q(\lambda)$ induces a structure on any categorical model similar to an initial object: there exists an object $Z$ such that for any object $A$ we have $\hom(Z^{\otimes m}, A)\not= \emptyset$ for some $m\geq 0$.
\end{example}

\appendix
\bibliography{all}

\end{document}